\documentclass[noinfoline,stslayout]{imsart}

\usepackage{amsfonts,amssymb,amsmath,amscd,amsthm,latexsym}
\usepackage{natbib}
\usepackage{graphicx}
\usepackage[]{units}
\usepackage{algorithmicx}
\usepackage{algpseudocode}
\usepackage{qtree}
\usepackage{algorithm}
\usepackage{xcolor}
\usepackage{comment}

\newtheorem{lemma}{Lemma}[section]
\newtheorem{theorem}{Theorem}[section]
\newtheorem{exemp}{Example}[section]

\newcommand\oh{\nicefrac{1}{2}}

\usepackage{amsthm}
\usepackage{bigints}
\usepackage{amsmath}
\usepackage{natbib}
\usepackage[colorlinks,citecolor=blue,urlcolor=blue,filecolor=blue,backref=page]{hyperref}
\usepackage{graphicx}
\usepackage{amsfonts,amssymb,amsmath,amscd,amsthm,latexsym}
\usepackage{natbib}
\usepackage{graphicx}
\usepackage[]{units}

\usepackage{algorithmicx}
\usepackage{algpseudocode}
\usepackage{qtree}
\usepackage{algorithm}
\usepackage{xcolor}
\usepackage{comment}

\startlocaldefs
\endlocaldefs

\begin{document}


\begin{frontmatter}
\title{Jeffreys priors for mixture estimation: properties and alternatives}

\author{ Clara Grazian\thanks
{Corresponging Author: Memotef,
Sapienza Universit\`a di Roma. 
Via del Castro Laurenziano, 9, 00161, Roma, Italy. CEREMADE Universit\'e Paris-Dauphine, Paris, France. e-mail: clara.grazian@uniroma1.it
} \and Christian P. Robert\thanks{CEREMADE Universit\'e Paris-Dauphine, University of Warwick and CREST, Paris.
e-mail: xian@ceremade.dauphine.fr. }
}

\begin{abstract}
While Jeffreys priors usually are well-defined for the parameters of mixtures of distributions, they are not available
in closed form. Furthermore, they often are improper priors. Hence, they have never been used to draw inference on the
mixture parameters. We study in this paper the implementation and the properties of Jeffreys priors in several 
mixture settings, show that the associated posterior distributions most often are improper, and then 
propose a noninformative alternative for the analysis of mixtures.
\end{abstract}

\begin{keyword}
\kwd{Noninformative prior}
\kwd{mixture of distributions}
\kwd{Bayesian analysis}
\kwd{Dirichlet prior}
\kwd{improper prior}
\kwd{improper posterior}
\kwd{label switching}
\end{keyword}

\end{frontmatter}

\section{Introduction}
\label{intro}

Bayesian inference in mixtures of distributions has been studied quite extensively in the literature. See, e.g.,
\cite{maclachlan:peel:2000} and \cite{fruhwirth:2006} for book-long references and
\cite{lee:marin:mengersen:robert:2008} for one among many surveys. From a Bayesian perspective, one of the several
difficulties with this type of distribution,
\begin{equation}\label{eq:theMix}
\sum_{i=1}^k p_i\,f(x|\theta_i)\,,\quad \sum_{i=1}^k p_i=1\,,
\end{equation}
is that its ill-defined nature (non-identifiability, multimodality, unbounded likelihood, etc.) 
leads to restrictive prior modelling since most
improper priors are not acceptable. This is due in particular to the feature that
a sample from \eqref{eq:theMix} may contain no subset from one of the
$k$ components $f(\cdot|\theta_i)$ (see. e.g.,
\citealp{titterington:smith:makov:1985}). Albeit the probability of such an event is
decreasing quickly to zero as the sample size grows, it nonetheless prevents the use
of independent improper priors, unless such events are prohibited \citep{diebolt:robert:1994}. 
Similarly, the exchangeable nature of the components often induces both multimodality in the posterior distribution and
convergence difficulties as exemplified by the {\em label switching} phenomenon that is now quite well-documented
\citep{celeux:hurn:robert:2000, stephens:2000b, jasra:holmes:stephens:2005, fruhwirth:2006, geweke:2007,
puolamaki:kaski:2009}. This feature is characterized by a lack of symmetry in the outcome of a Monte Carlo Markov chain
(MCMC) algorithm, in that the posterior density is exchangeable in the components of the mixture but the MCMC sample
does not exhibit this symmetry. In addition, most MCMC samplers do not concentrate around a single mode of the posterior
density, partly exploring several modes, which makes the construction of Bayes estimators of the components much harder.

When specifying a prior over the parameters of \eqref{eq:theMix}, it is
therefore quite delicate to produce a manageable and sensible non-informative
version and some have argued against using non-informative priors
in this setting (for example, \cite{maclachlan:peel:2000} argue that it is impossible to obtain proper posterior distribution from fully noninformative priors), on the basis that mixture models were ill-defined objects that required informative priors to give a meaning to the notion of a component of
\eqref{eq:theMix}. For instance, the distance between two components needs to be
bounded from below to avoid repeating the same component over and over again.
Alternatively, the components all need to be informed by the data, as
exemplified in \cite{diebolt:robert:1994} who imposed a completion scheme
(i.e., a joint model on both parameters and latent variables) such that all
components were allocated at least two observations, thereby ensuring that the
(truncated) posterior was well-defined. \cite{wasserman:2000} proved ten years
later that this truncation led to consistent estimators and moreover that only
this type of priors could produce consistency. While the constraint on the
allocations is not fully compatible with the i.i.d. representation of a mixture
model, it naturally expresses a modelling requirement that all components have
a meaning in terms of the data, namely that all components genuinely
contributed to generating a part of the data. This translates as a form of weak
prior information on how much one trusts the model and how meaningful each
component is on its own (by opposition with the possibility of adding
meaningless artificial extra-components with almost zero weights or almost
identical parameters).

While we do not seek Jeffreys priors as the ultimate prior modelling for non-informative settings, being altogether
convinced of the lack of unique reference priors \citep{robert:2001,robert:chopin:rousseau:2009}, we think it is
nonetheless worthwile to study the performances of those priors in the setting of mixtures in order to determine if
indeed they can provide a form of reference priors and if they are at least well-defined in such settings. We will show that only in very specific situations the Jeffreys prior provides reasonable inference.

In Section \ref{sec:jeffreys} we provide a formal characterisation of properness of the posterior distribution for the parameters of a mixture model, in particular with Gaussian components, when a Jeffreys prior is used for them. In Section \ref{sec:prosper} we will analyze the properness of the Jeffreys prior and of the related posterior distribution: only when the weights of the components (which are defined in a compact space) are the only unknown parameters it turns out that the Jeffreys prior (and so the relative posterior) is proper; on the other hand, when the other parameters are unknown, the Jeffreys prior will be proved to be improper and in only one situation it provides a proper posterior distribution. In Section \ref{sec:alternative} we propose a way to realize a noninformative analysis of mixture models and introduce improper priors for at least some parameters. Section \ref{sec:concl} concludes the paper. 

\vspace{0.3cm}
\section{Jeffreys priors for mixture models}
\label{sec:jeffreys}

We recall that the Jeffreys prior was introduced by \cite{jeffreys:1939} as a default prior based
on the Fisher information matrix
$$
\pi^\text{J}(\theta) \propto |I(\theta)|^{\oh}\,,
$$
whenever the later is well-defined; $I(\cdot)$ stand for the expected Fisher information matrix and the symbol $|\cdot|$
denotes the determinant. Although the prior is endowed with some frequentist properties like matching and asymptotic
minimal information \citep[Chapter 3]{robert:2001}, it does not constitute the ultimate answer to the selection of prior
distributions in non-informative settings and there exist many alternative such as reference priors
\citep{berger:bernardo:sun:2009}, maximum entropy priors \citep{rissanen:2012}, matching priors
\citep{ghosh:carlin:srivastava:1995}, and other proposals \citep{kass:wasserman:1996}. In most settings Jeffreys priors
are improper, which may explain for their conspicuous absence in the domain of mixture estimation, since the latter
prohibits the use of most improper priors by allowing any subset of components to go ``empty" with positive probability.
That is, the likelihood of a mixture model can always be decomposed as a sum over all possible partitions of the data
into $k$ groups at most, where $k$ is the number of components of the mixture. This means that there are terms in this
sum where no observation from the sample brings any amount of information about the parameters of a specific component. 

Approximations of the Jeffreys prior in the setting of mixtures can be found, e.g., in
\cite{figueiredo:jain:2002}, where the Authors revert to independent Jeffreys priors on the components of the mixture. This induces the same negative side-effect as with other independent priors, namely an impossibility to handle improper
priors.

\cite{rubio:steel:2014} provide a closed-form expression for the Jeffreys prior for a location-scale mixture with two components. The family of distributions they consider is
$$
\dfrac{2\epsilon}{\sigma_1}f\left(\frac{x-\mu}{\sigma_1}\right)\mathbb{I}_{x<\mu}+
\dfrac{2(1-\epsilon)}{\sigma_2}f\left(\frac{x-\mu}{\sigma_2}\right) \mathbb{I}_{x>\mu}
$$
(which thus hardly qualifies as a mixture, due to the orthogonality in the supports of both components that allows to identify which component each observation is issued from). The factor $2$ in the fraction
is due to the assumption of symmetry 
around zero for the density $f$. For this specific model, if we impose that the weight $\epsilon$ is a function of the variance parameters,
$
\epsilon=\nicefrac{\sigma_1}{\sigma_1+\sigma_2},
$
the Jeffreys prior is given by
$
\pi(\mu,\sigma_1,\sigma_2) \propto \nicefrac{1}{\sigma_1\sigma_2\{\sigma_1+\sigma_2\}}.
$
However, in this setting, \cite{rubio:steel:2014} demonstrate that the posterior associated with the (regular)
Jeffreys prior is improper, hence not relevant for conducting inference. (One may wonder at the pertinence of
a Fisher information in this model, given that the likelihood is not differentiable in $\mu$.)
\cite{rubio:steel:2014} also consider alternatives to the genuine Jeffreys prior, either by reducing the range or even
the number of parameters, or by building a product of conditional priors. They further consider so-called non-objective priors that are only relevant to the specific
case of the above mixture.

Another obvious explanation for the absence of Jeffreys priors is computational, namely the closed-form
derivation of the Fisher information matrix is almost inevitably impossible.
The reason is that integrals of the form

\begin{equation*}
-\int_{\mathcal{X}} \frac{\partial^2 \log \left[\sum_{h=1}^k p_h\,f(x|\theta_h)\right]}{\partial \theta_i \partial \theta_j}\left[\sum_{h=1}^k p_h\,f(x|\theta_h)\right]^{-1} d x
\end{equation*}

\noindent (in the special case of component densities with a single parameter) cannot be computed analytically.
We derive an approximation of the elements of the Fisher information matrix based on Riemann sums. 
The resulting computational expense is of order $\mathrm{O}(d^2)$ if $d$ is the total number of (independent) parameters.
Since the elements of the information matrix usually are ratios between the component densities and the mixture
density, there may be difficulties with non-probabilistic methods of integration. Here, we use Riemann sums (with $550$
points) when the component standard deviations are sufficiently large, as they produce stable results, and Monte Carlo 
integration (with sample sizes of $1500$) when they are small. In the latter case, the variability of MCMC results seems 
to decrease as $\sigma_i$ approaches $0$. 

\vspace{0.3cm}
\section{Properness for prior and posterior distributions}
\label{sec:prosper}

Unsurprisingly, most Jeffreys priors associated with mixture models are improper, the exception being when only the
weights of the mixture are unknown, as already demonstrated in \cite{bernardo:giron:1988}. 

We will characterize properness and improperness of Jeffreys priors and derived posteriors, when some or all of the parameters of distributions from location-scale families are unknown. These results are established both analytically and via simulations, with sufficiently large Monte Carlo experiments checking the behavior of the approximated posterior distribution.

\subsection{Characterization of Jeffreys priors}
\label{subsec:priors}

\subsubsection{Weights of mixture unknown}

A representation of the Jeffreys prior and the derived posterior distribution for the weights of a 3-component mixture
model is given in Figure \ref{weights-priorpost}: the prior distribution is much more concentrated around extreme
values in the support, i.e., it is a prior distribution conservative in the number of important components. 

\begin{figure}
\centering
\includegraphics[width=6.5cm, height=7.5cm]{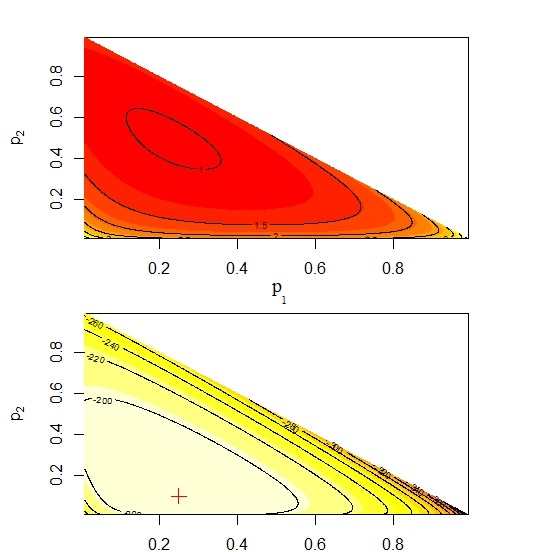}
\caption{Approximations (on a grid of values) of the Jeffreys prior (on the log-scale) when only the weights of a
Gaussian mixture model with 3-components are unknown (on the top) and of the derived posterior distribution (with known
means equal to -1, 0 and 2 respectively and known standard devitations equal to 1, 5 and 0.5 respectively). The red
cross represents the true values.} 
\label{weights-priorpost}
\end{figure}

\begin{lemma} 
\label{lem:weights}
When the weights $p_i$ are the only unknown parameters in \eqref{eq:theMix}, the corresponding Jeffreys prior 
is proper. 
\end{lemma}

Figure \ref{weights-boxplots} shows the boxplots for the means of the approximated posterior distribution for the weights of a three-component Gaussian mixture model.

\begin{figure}
\centering
\includegraphics[width=6.5cm, height=7.5cm]{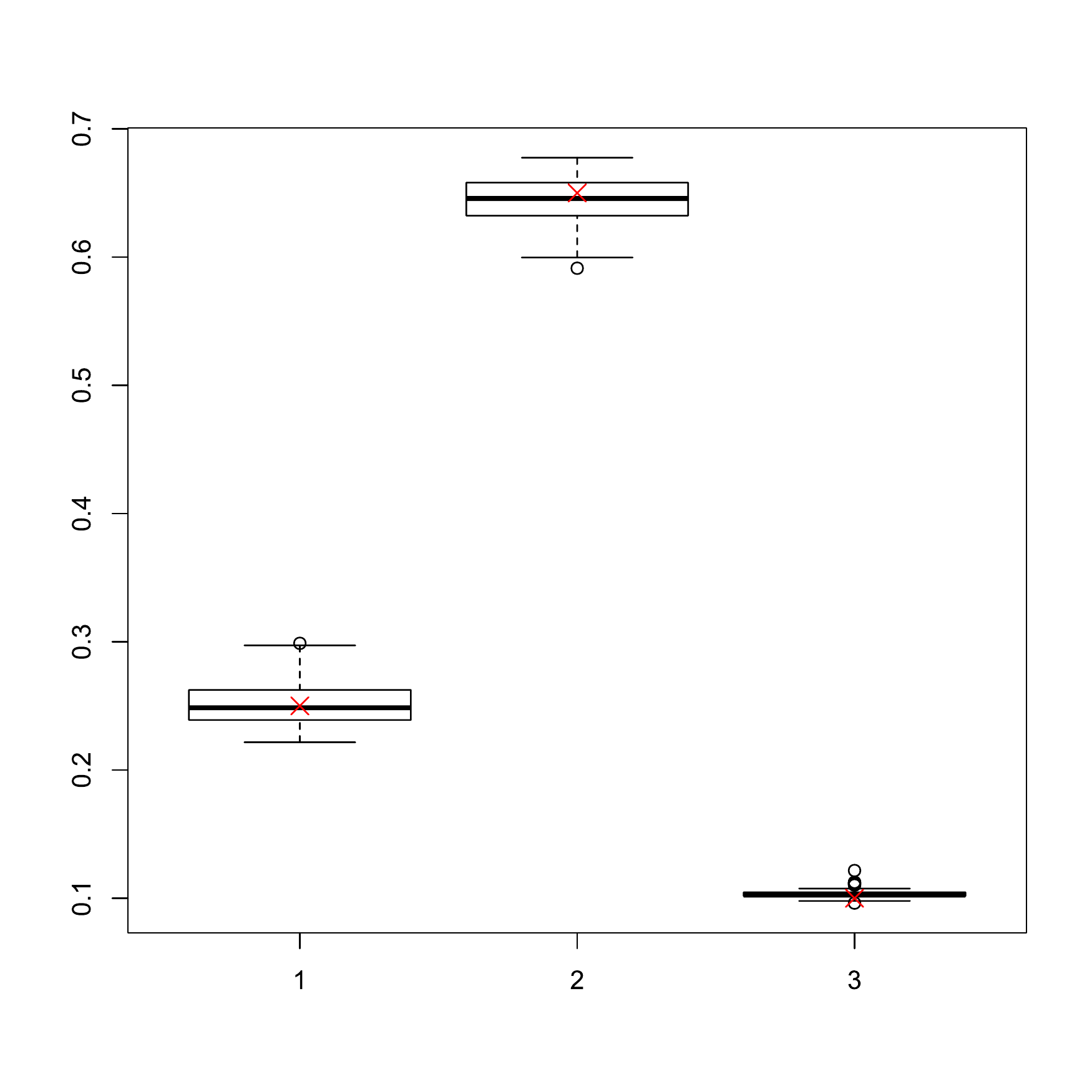}
\caption{Boxplots of the estimated means of the three-component mixture model
$0.25\mathcal{N}(-10,1)+0.65\mathcal{N}(0,5) +0.10\mathcal{N}(15,0.5)$ for 50 simulated samples of size $100$, obtained
via MCMC with $10^5$ simulations. The red crosses represent the true values of the weights.}
\label{weights-boxplots}
\end{figure}

\begin{proof}
The generic element of the Fisher information matrix is (for $i,j=\{1,\ldots,k-1\}$)

\begin{equation}
\int_\mathcal{X} \frac{(f_i(x)-f_k(x))(f_j(x)-f_k(x))}{\sum_{l=1}^k p_l f_l(x)}
d x
\label{eq:ww-prior}
\end{equation}
when we consider the parametrization in $(p_1,\ldots,p_{k-1})$, with
$$
p_k=1-p_1-\cdots-p_{k-1}\,.
$$

We remind that, since the Fisher information matrix is a positive semi-definite, its determinant is bounded by the product of the terms in the diagonal, thanks to the Hadamard's inequality.
 Therefore, we may consider the diagonal term,
\begin{align*}
\int_\mathcal{X} \frac{(f_i(x)-f_k(x))^2}{\sum\limits_{l=1}^k p_l f_l(x)} d x
&= \int_{f_i(x)\ge f_k(x)} \frac{(f_i(x)-f_k(x))^2}{\sum\limits_{l=1}^k p_l f_l(x)} d x\\
&\quad + \int_{f_i(x)\le f_k(x)} \frac{|(f_i(x)-f_k(x))^2|}{\sum\limits_{l=1}^k p_l f_l(x)} d x\\
&= \int_{f_i(x)\ge f_k(x)} \frac{f_i(x)-f_k(x)}{\sum\limits_{l=1}^k p_l f_l(x)} \{f_i(x)-f_k(x)\}d x\\
&\quad + \int_{f_i(x)\le f_k(x)} \Big| \frac{f_i(x)-f_k(x)}{\sum\limits_{l=1}^k p_l f_l(x)} \Big| |f_i(x)-f_k(x)| d x\\
&= \frac{1}{p_i}\,\int_{f_i\ge f_k} \frac{p_i\{f_i(x)-f_k(x)\}}{p_i\{f_i(x)-f_k(x)\}+\sum\limits_{l\ne i,k} p_l
\{f_l(x)-f_k(x)\}+f_k(x)}\\
&\qquad\qquad \{f_i(x)-f_k(x)\}d x\\
&\quad + \frac{1}{p_i}\,\int_{f_i\le f_k} \Big|\frac{p_i\{f_i(x)-f_k(x)\}}{p_i\{f_i(x)-f_k(x)\}+\sum\limits_{l\ne i,k} p_l
\{f_l(x)-f_k(x)\}+f_k(x)} \Big| \\
&\qquad\qquad |f_i(x)-f_k(x)| d x\\
&\le \frac{1}{p_i}\int_{f_i(x)\ge f_k(x)} \{f_i(x)-f_k(x)\}d x + \frac{1}{p_i}\int_{f_i(x)\le f_k(x)} | f_i(x)-f_k(x) |d x\\
&= \frac{2}{p_i}\int_{f_i(x)\ge f_k(x)} \{f_i(x)-f_k(x)\}d x
\end{align*}
since both integrals are equal.

Therefore, the Jeffreys prior will be bounded by the square root of the product of the terms in the diagonal of the Fisher information matrix

\begin{equation*}
\pi^J(\mathbf{p}) \propto \prod_{i=1}^k p_i^{-\frac{1}{2}}
\end{equation*}

\noindent which is a generalization to $k$ components of the prior provided in \cite{bernardo:giron:1988} for $k=2$ (however, \cite{bernardo:giron:1988} find the reference prior for the limiting case when all the components have pairwise disjoint supports, while for the opposite limiting case where all the components converge to the same distribution, the Jeffrey's prior is the uniform distribution on the $k$-dimensional simplex). 
\end{proof}

This reasoning leads \cite{bernardo:giron:1988} to conclude that the usual $\mathcal{D}(\lambda_1,\ldots,\lambda_k)$
Dirichlet prior with $\lambda_i \in [\nicefrac{1}{2},1]$ for $\forall i=1,\cdots,k$ seems to be a reasonable
approximation. They also prove that the Jeffreys prior for the weights $p_i$ is convex, with a argument based on the sign of the second derivative. 

As a remark, the configuration shown in proof of Lemma \ref{lem:weights} is compatible with the Dirichlet configuration of the prior proposed by \cite{rousseau:mengersen:2011}. 

The shape of the Jeffreys prior for the weights of a mixture model depends on the type of the components. 
Figure \ref{weights-GMM}, \ref{weights-GtMM} and \ref{weights-GtMM-df} show the form of the Jeffreys prior for a
2-component mixture model for different choices of components. It is always concentrated around the extreme values of the
support, however the amount of concentration around $0$ or $1$ depends on the information brought by each
component. In particular, Figure \ref{weights-GMM} shows that the prior is much more symmetric as there is symmetry
between the variances of the distribution components, while Figure \ref{weights-GtMM} shows that the prior is much more
concentrated around 1 for the weight relative to the normal component if the second component is a Student t
distribution. 

Finally Figure \ref{weights-GtMM-df} shows the behavior of the Jeffreys prior when the first component is Gaussian and the second is a Student t and the number of degrees of freedom is increasing. As expected, as the Student t is approaching a normal distribution, the Jeffreys prior becomes more and more symmetric.

\begin{figure}
\centering
\includegraphics[width=6.5cm, height=7.5cm]{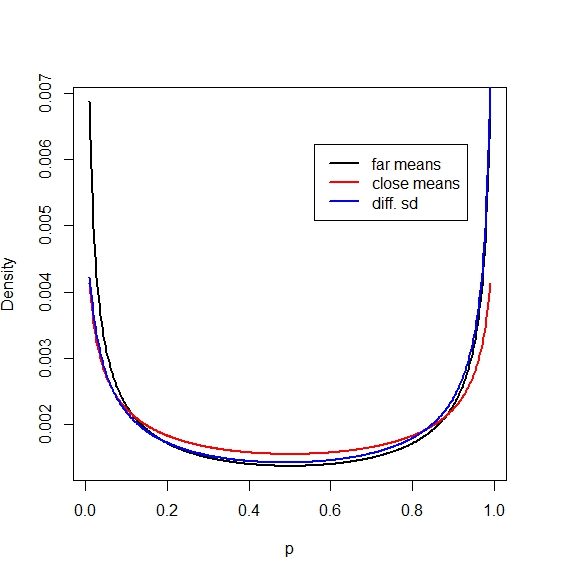}
\caption{Approximations of the marginal prior distributions for the first weight of a 2-component Gaussian mixture model, $p\,\mathcal{N}(-10,1)+(1-p)\,\mathcal{N}(10,1)$ (black), $p\,\mathcal{N}(-1,1)+(1-p)\,\mathcal{N}(1,1)$ (red) and $p\,\mathcal{N}(-10,1)+(1-p)\,\mathcal{N}(10,10)$ (blue).}
\label{weights-GMM}
\end{figure}

\begin{figure}
\centering
\includegraphics[width=6.5cm, height=7.5cm]{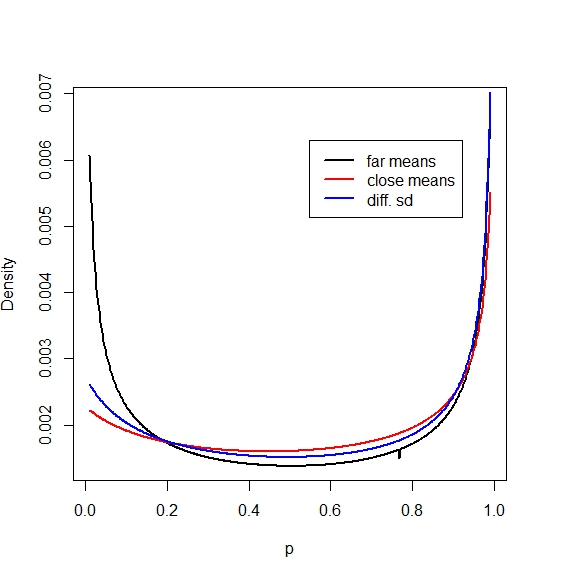}
\caption{Approximations of the marginal prior distributions for the first weight of a 2-component mixture model where the first component is Gaussian and the second is Student t, $p\,\mathcal{N}(-10,1)+(1-p)\,\mathrm{t}(df=1,10,1)$ (black), $p\,\mathcal{N}(-1,1)+(1-p)\,\mathrm{t}(df=1,1,1)$ (red) and $p\,\mathcal{N}(-10,1)+(1-p)\,\mathrm{t}(df=1,10,10)$ (blue).}
\label{weights-GtMM}
\end{figure}

\begin{figure}
\centering
\includegraphics[width=6.5cm, height=7.5cm]{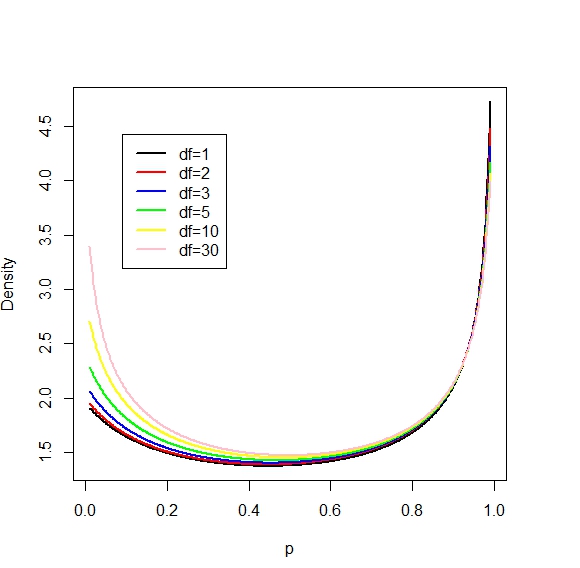}
\caption{Approximations of the marginal prior distributions for the first weight of a 2-component mixture model where the first component is Gaussian and the second is Student t with an increasing number of degrees of freedom.}
\label{weights-GtMM-df}
\end{figure}

\subsubsection{Location and scale parameters of a mixture model unknown}

If the components of the mixture model \eqref{eq:theMix} are distributions from a location-scale family and the location or scale parameters of the mixture components are unknown, this turns the mixture itself into a location-scale model. As a result, model \eqref{eq:theMix} may be reparametrized by following \cite{mengersen:robert:1996}, in the case of Gaussian components

\begin{equation}
\label{reparMix}
p\mathcal{N}(\mu,\tau^2)+(1-p)\mathcal{N}(\mu+\tau\delta,\tau^2\sigma^2)
\end{equation}

\noindent namely using a reference location $\mu$ and a reference scale $\tau$ (which may be, for instance, the location and scale of a specific component). Equation \eqref{reparMix} may be generalized to the case of $k$ components as

\begin{align}
\label{eq:k_reparMix}
p\mathcal{N}(\mu,\tau^2)&+\sum_{i=1}^{k-2} (1-p) (1-q_1) \cdots (1-q_{i-1})q_i \mathcal{N}(\mu+\tau\theta_1+\cdots+\tau\cdots\sigma_{i-1}\theta_i,\tau^2\sigma_1^2\cdots\sigma_i^2) \nonumber \\
&\qquad {} + (1-p)(1-q_1)\cdots (1-q_{k-2})\mathcal{N}(\mu+\tau\theta_1+\cdots+\tau\cdots\sigma_{k-2}\theta_{k-1},\tau^2\sigma_1^2\cdots\sigma_{k-1}^2)
\end{align}

In this way, the mixture model is more cleary a location-scale model, which implies that the Jeffreys prior is flat in the location and powered as $\tau^{-d/2}$ if $d$ is the total number of parameters of the components, respectively \citep[Chapter 3]{robert:2001}, as we will see in the following.

\begin{lemma}\label{lem:meansd-prior}
If the parameters of the components of a mixture model are either location or scale parameters, the corresponding Jeffreys prior is improper. 
\end{lemma}

In the proof of Lemma \ref{lem:meansd-prior}, we will consider a Gaussian mixture model and then extend the results to the general situation of components from a location-scale family.
 
\subsubsection*{Unknown location parameters}

\begin{proof}
We first consider the case where the means are the only unknown parameters of a Gaussian mixture model

\begin{equation*}
g_X(x)=\sum_{l=1}^k p_l \mathfrak{n}(x|\mu_l,\sigma_l^2)
\end{equation*}

The generic elements of the expected Fisher information matrix are, in the case of diagonal and off-diagonal terms respectively:

$$
\mathbb{E}\left[- \frac{\partial^2 \log g_X(X)}{\partial \mu_i^2}\right]=\frac{p_i^2}{\sigma_i^4} \bigintsss_{-\infty}^\infty
\frac{\left[ (x-\mu_i) \mathfrak{n}(x|\mu_i,\sigma_i^2)\right]^2}{\sum_{l=1}^k p_l \mathfrak{n}(x|\mu_l,\sigma_l^2)}
d x
$$

$$
\mathbb{E}\left[- \frac{\partial^2 \log g_X(X)}{\partial \mu_i \partial \mu_j}\right]=\frac{p_i p_j}{\sigma_i^2
\sigma_j^2} \bigintsss_{-\infty}^\infty \frac{(x-\mu_i) \mathfrak{n}(x|\mu_i,\sigma_i^2)(x-\mu_j)
\mathfrak{n}(x|\mu_j,\sigma_j^2) }{\sum_{l=1}^k p_l \mathfrak{n}(x|\mu_l,\sigma_l^2)}  d x
$$

Now, consider the change of variable $t=x-\mu_i$ in the above integrals, where $\mu_i$ is thus the mean of the $i$-th Gaussian component
($i\in\{1,\cdots,k\}$). The above integrals are then equal to
\begin{align*}
\mathbb{E}\left[- \frac{\partial^2 \log g_X(X)}{\partial \mu_j^2}\right] &= \frac{p_j^2}{\sigma_j^4} \bigintsss_{-\infty}^\infty
\frac{\left[ (t-\mu_j+\mu_i) \mathfrak{n}(t|\mu_j-\mu_i,\sigma_i^2)\right]^2}{\sum_{l=1}^k p_l
\mathfrak{n}(t|\mu_l-\mu_i,\sigma_l^2)} d x\\
\mathbb{E}\left[- \frac{\partial^2 \log g_X(X)}{\partial \mu_j \partial \mu_m}\right] &= \frac{p_j p_m}{\sigma_j^2
\sigma_m^2} \bigintsss_{-\infty}^\infty \frac{(t-\mu_j+\mu_i) \mathfrak{n}(x|\mu_j,\sigma_j^2)(t-\mu_m+\mu_i)
\mathfrak{n}(t|\mu_m-\mu_i,\sigma_m^2) }{\sum_{l=1}^k p_l \mathfrak{n}(t|\mu_l-\mu_i,\sigma_l^2)}  d x\\
\label{eq:means-prior}
\end{align*}
Therefore, the terms in the Fisher information only depend on the differences $\delta_j=\mu_i-\mu_j$ for $j \in
\{1,\cdots,k \}$. This implies that the Jeffreys prior is improper since a reparametrization in
($\mu_i,\mathbf{\delta}$) shows the prior does not depend on $\mu_i$.

This feature will reappear whenever the location parameters are unknown.

When considering the general case of components from a location-scale family, this feature of improperness of the Jeffreys prior distribution is still valid, because, once reference location-scale parameters are chosen, the mixture model may be rewritten as

\begin{equation}
\label{eq:mix-locscale}
p_1 f_1(x|\mu,\tau)+\sum_{i=2}^k p_i f_i(\frac{a_i+ x}{b_i} |\mu,\tau,a_i,b_i).
\end{equation}

Then the second derivatives of the logarithm of model \eqref{eq:mix-locscale} behave as the ones we have derived for the Gaussian case, i.e. they will depend on the differences between each location parameter and the reference one, but not on the reference location itself. Then the Jeffreys prior will be constant with respect to the global location parameter. 

\end{proof}

When considering the reparametrization \eqref{reparMix}, the Jeffreys prior for $\delta$ for a fix $\mu$ has the form:

\begin{equation*}
\pi^J(\delta|\mu)\propto \left[
\int_\mathfrak{X}\frac{\left[{(1-p)x\exp\{-\frac{x^2}{2}\}}\right]^2}{{p\sigma\exp\{-\frac{\sigma^2(x+\frac{\delta}{\sigma\tau})^2}{2}\}}+{(1-p)\exp\{-\frac{x^2}{2}\}}}
d x \right]^{\frac{1}{2}}
\end{equation*}

\noindent and the following result may be demonstrated.

\begin{lemma} 
The Jeffreys prior of $\delta$ conditional on $\mu$ when only the location parameters are unknown is improper.
\end{lemma}

\begin{proof}
The improperness of the conditional Jeffreys prior on $\delta$ depends (up to a constant) on the double integral

\begin{eqnarray*}
\int_\Delta
\int_\mathfrak{X} c \frac{\left[(1-p)x\exp\{-\frac{x^2}{2}\}\right]^2}{p\sigma\exp\{-\frac{\sigma^2(x+\frac{\delta}{\sigma\tau})^2}{2}\}+(1-p)\exp\{-\frac{x^2}{2}\}}
d x d\delta.
\end{eqnarray*}

\noindent The order of the integrals is allowed to be changed, then

\begin{eqnarray*}
\int_\mathfrak{X} x^2 \int_\Delta
\frac{\left[(1-p)\exp\{-\frac{x^2}{2}\}\right]^2}{p\sigma\exp\{-\frac{\sigma^2(x+\frac{\delta}{\sigma\tau})^2}{2}\}+(1-p)\exp\{-\frac{x^2}{2}\}}
d\delta d x 
\end{eqnarray*}

\noindent Define $f(x)=(1-p)e^{-\frac{x^2}{2}}=\frac{1}{d}$. Then

\begin{eqnarray*}
\int_\mathcal{X} x^2 \int_\Delta \frac{1}{d^2
p\sigma\exp\{-\frac{\sigma^2(x+\frac{\delta}{\sigma\tau})^2}{2}\}+d} d\delta d x 
\end{eqnarray*}

\noindent Since the behavior of $\left[d^2
p\sigma\exp\{-\frac{\sigma^2(x+\frac{\delta}{\sigma\tau})^2}{2}\}+d\right]$ depends on $\exp\{-\delta^2\}$
as $\delta$ goes to $\infty$, we have that 

\begin{equation*}
\int_{-\infty} ^{+\infty} \frac{1}{\exp\{-\delta^2\}+d} d\delta > \int_{A} ^{+\infty} \frac{1}{\exp\{-\delta^2\}+d} d\delta  
\end{equation*}

\noindent because the integrand function is positive. Then

\begin{equation*}
\int_{A} ^{+\infty} \frac{1}{\exp\{-\delta^2\}+d} d\delta > \int_{A} ^{+\infty} \frac{1}{\varepsilon+d} d\delta = +\infty
\end{equation*}

Therefore the conditional Jeffreys prior on $\delta$ is improper.

Figure \ref{fig:priorpost-diff} compares the
behavior of the prior and the resulting posterior distribution for the difference between the means of a two-component
Gaussian mixture model: the prior distribution is symmetric and it has different behaviors depending on the value of the
other parameters, but it always stabilizes for large enough values; the posterior distribution appears to always concentrate
around the true value. 

\begin{figure}
\centering
\includegraphics[width=6.5cm, height=7.5cm]{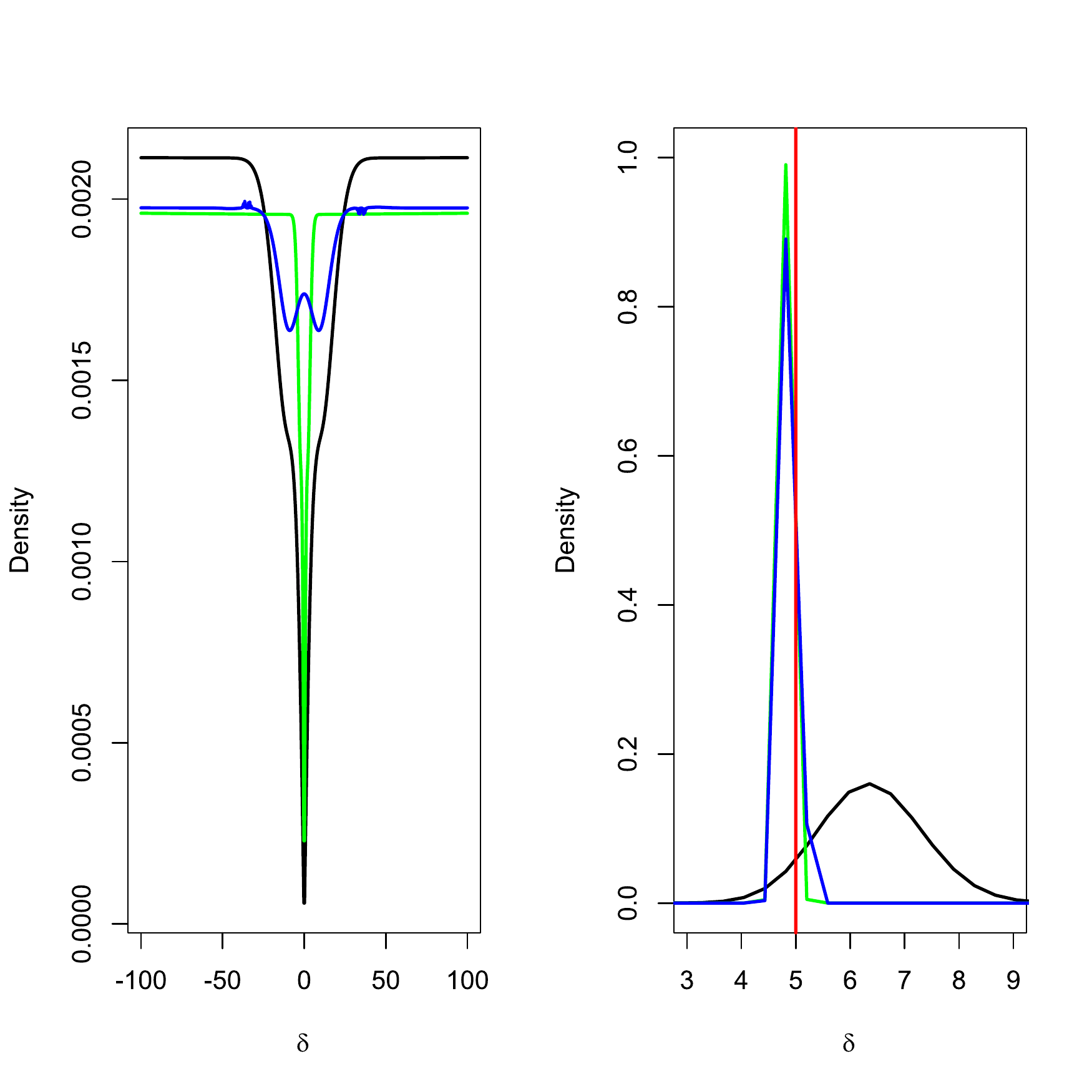}
\caption{Approximations (on a grid of values) of the Jeffreys prior (on the natural scale) of the difference between the means of a Gaussian mixture model with only the means unknown (left) and of the derived posterior distribution (on the right, the red line represents the true value), with known weights equal to $(0.5,0.5)$ (black lines), $(0.25,0.75)$ (green and blue lines) and known standard deviations equal to $(5,5)$ (black lines), $(1,1)$ (green lines) and $(7,1)$ (blue lines).} 
\label{fig:priorpost-diff}
\end{figure} 

\end{proof}

\subsubsection*{Unknown scale parameters}

Consider now the second case of the scale parameters being the only unknown parameters. 

\begin{proof}
First, consider a Gaussian mixture model and suppose the mixture model is composed by only two components; the Jeffreys prior for the scale parameters is defined as

\begin{align*}
\pi^J(\sigma_1,\sigma_2)&\propto \left\{ \frac{p_1^2}{\sigma_1^2} \bigintsss_{-\infty}^\infty
\frac{\left[ \left(\frac{(x-\mu_1)^2}{\sigma_1^2}-1\right) \mathfrak{n}(x|\mu_1,\sigma_1^2)\right]^2}{\sum_{l=1}^2 p_l \mathfrak{n}(x|\mu_l,\sigma_l^2)}
d x \right. \nonumber \\
					&\cdot \left.{} \frac{p_2^2}{\sigma_2^2} \bigintsss_{-\infty}^\infty
\frac{\left[ \left(\frac{(x-\mu_2)^2}{\sigma_2^2}-1\right) \mathfrak{n}(x|\mu_2,\sigma_2^2)\right]^2}{\sum_{l=1}^2 p_l \mathfrak{n}(x|\mu_l,\sigma_l^2)}
d x \right. \nonumber \\
					&- \left.{} \left[\frac{p_1 p_2}{\sigma_1 \sigma_2} \bigintsss_{-\infty}^\infty
\frac{ \left(\frac{(x-\mu_1)^2}{\sigma_1^2}-1\right) \left(\frac{(x-\mu_2)^2}{\sigma_2^2}-1\right) \mathfrak{n}(x|\mu_1,\sigma_1^2)\mathfrak{n}(x|\mu_2,\sigma_2^2)}{\sum_{l=1}^2 p_l \mathfrak{n}(x|\mu_l,\sigma_l^2)} d x \right]^2\right\}^\frac{1}{2} 
\end{align*}

Since the Fisher information matrix is positive definite, it is bounded by the product on the diagonal, then we can write:

\begin{align*}
\pi^J(\sigma_1,\sigma_2)&\leq c \frac{p_1 p_2}{\sigma_1\sigma_2}\left\{ \bigintsss_{-\infty}^\infty
\frac{\left(\frac{(x-\mu_1)^2}{\sigma_1^2}-1\right)^2 \frac{1}{\sigma_1^2} \exp\left\{ -\frac{(x-\mu_1)^2}{\sigma_1^2}\right\}}{\frac{p_1}{\sigma_1}\exp\left\{-\frac{(x-\mu_1)^2}{2\sigma_1^2}\right\}+\frac{p_2}{\sigma_2}\exp\left\{-\frac{(x-\mu_2)^2}{2\sigma_2^2}\right\}}
d x \right. \nonumber \\
					&\cdot \left.{} \bigintsss_{-\infty}^\infty
\frac{\left(\frac{(x-\mu_2)^2}{\sigma_2^2}-1\right)^2 \frac{1}{\sigma_2^2} \exp\left\{ -\frac{(x-\mu_2)^2}{\sigma_2^2}\right\}}{\frac{p_1}{\sigma_1}\exp\left\{-\frac{(x-\mu_1)^2}{2\sigma_1^2}\right\}+\frac{p_2}{\sigma_2}\exp\left\{-\frac{(x-\mu_2)^2}{2\sigma_2^2}\right\}}
d x \right\}^\frac{1}{2} 
\end{align*}

In particular, if we reparametrize the model by introducing $\sigma_1=\tau$ and $\sigma_2=\tau \sigma$ and study the behavior of the following integral

\begin{align}
\label{eq:scaleprior}
\bigintsss_0^{\infty} \bigintsss_0^{\infty} & c \frac{p_1 p_2}{\tau\sigma}\left\{ \bigintsss_{-\infty}^\infty
\frac{\left(z^2-1\right)^2 \exp\left\{ -z^2\right\}}{p_1 \exp\left\{-\frac{z^2}{2}\right\}+\frac{p_2}{\sigma}\exp\left\{-\frac{(z\tau+\mu_1-\mu_2)^2}{2\tau^2\sigma^2}\right\}}
d z \right. \nonumber \\
					&\cdot \left.{} \left\{ \bigintsss_{-\infty}^\infty
\frac{\left(u^2-1\right)^2 \exp\left\{-u^2\right\}}{p_1\sigma \exp\left\{-\frac{(u\tau\sigma+\mu_2-\mu_1)^2}{2\tau^2}\right\}+p_2\exp\left\{-\frac{u^2}{2}\right\}}\right\} 
d u \right\}^\frac{1}{2} d \tau d \sigma
\end{align}
\noindent where the internal integrals with respect to $z$ and $u$ converge with respect to $\sigma$ and $\tau$, then the behavior of the external
integrals only depends on $\frac{1}{\tau\sigma}$. Therefore they do not converge. 

This proof can be easily extended to the case of $k$ components: the behavior of the prior depends on the inverse of the
product of the scale parameters, which implies that the prior is improper. 

Moreover this proof may be easily extended to the general case of mixtures of location-scale distributions \eqref{eq:mix-locscale}, because the second derivatives of the logarithm of the model will depend on factors $b_i^{-2}$ for $i \in {1,\cdots,k}$. When the square root is considered, it is evident that the integral will not converge. 
\end{proof}

Figures \ref{fig:sd-priorpost-clm} and \ref{fig:sd-priorpost-asym} show the prior and the posterior distributions of the scale parameters of a two-component mixture model for some situations with different weights and different means.

\begin{figure}
\centering
\includegraphics[width=6.5cm, height=7.5cm]{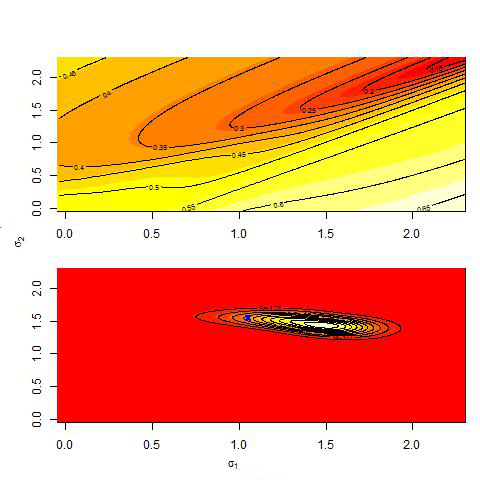}
\caption{Same as Figure \ref{fig:sd-priorpost-farm} but with known weights equal to $(0.25,0.75)$ and known means equal to $(-1,1)$.}
\label{fig:sd-priorpost-clm}
\end{figure}

\begin{figure}
\centering
\includegraphics[width=6.5cm, height=7.5cm]{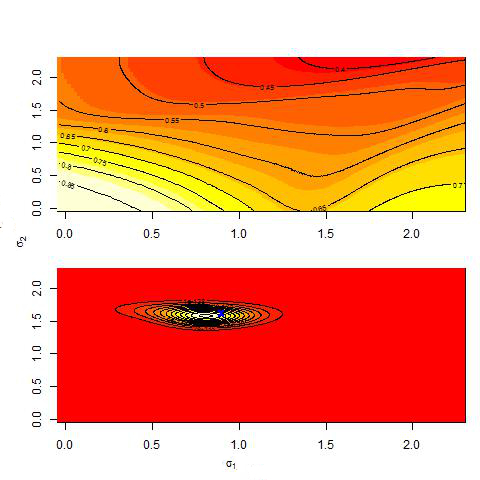}
\caption{Same as Figure \ref{fig:sd-priorpost-farm} but with known weights equal to $(0.25,0.75)$ and known means equal to $(-2,7)$.}
\label{fig:sd-priorpost-asym}
\end{figure}

Summarized results of the posterior approximation obtained via a random-walk Metropolis-Hastings algorithm by exploring
the posterior distribution associated with the Jeffreys prior on the standard deviations are shown in Figures
\ref{fig:sd2-bxp} and \ref{fig:sd3-bxp}, which display boxplots of the posterior means: provided a sufficiently high
sample size, simulations exhibit a convergent behavior. 

\begin{figure}
\centering
\includegraphics[width=6.5cm, height=7.5cm]{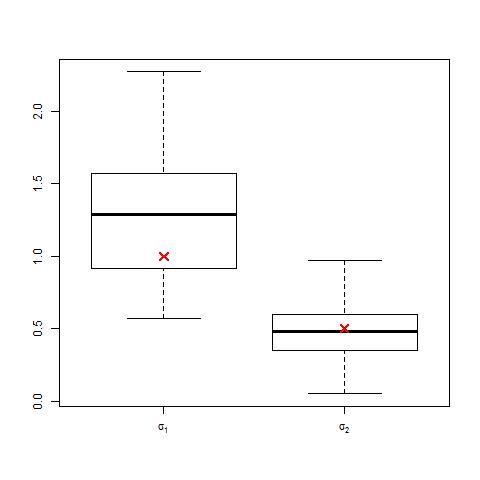}
\caption{Boxplots of posterior means of the standard deviations of the two-component mixture model $0.50\mathcal{N}(-1,1) + 0.50\mathcal{N}(2,0.5)$ for 50 replications of the experiment and a sample size equal to $10$, obtained via MCMC with $10^5$ simulations. The red cross represents the true values.}
\label{fig:sd2-bxp}
\end{figure} 

\begin{figure}
\centering
\includegraphics[width=6.5cm, height=7.5cm]{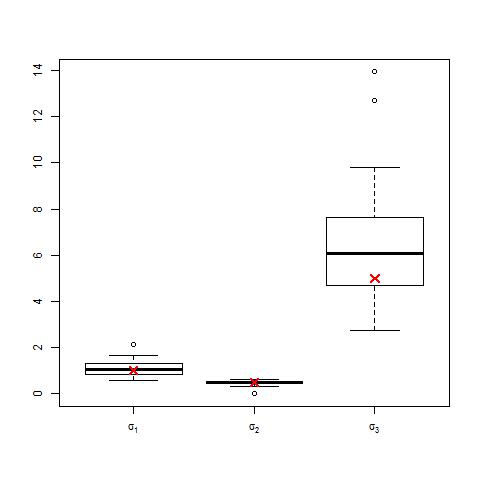}
\caption{Boxplots of posterior means of the standard deviations of the three-component mixture model $0.25\mathcal{N}(-1,1) + 0.65\mathcal{N}(0,0.5) + 0.10\mathcal{N}(2,5)$ for 50 replications of the experiment and a sample size equal to $50$, obtained via MCMC with $10^5$ simulations. The red cross represents the true values.}
\label{fig:sd3-bxp}
\end{figure} 

\subsubsection{Location and scale parameters unknown.}

\begin{proof}

Consider now the case where both location and scale parameters are unknown. Once again, each element of the Fisher
information matrix is an integral in which a change of variable $x-\mu_i$ can be used, for some choice of
$\mu_i,\ ,i=1,\cdots,k$ so that each term only depends on the difference $\delta_j=\mu_i-\mu_j$; the elements are
$$
\mathbb{E}\left[- \frac{\partial^2 log f_X(X)}{\partial \sigma_i^2}\right]=\frac{p_i^2}{\sigma_i^2} \int_{-\infty}^\infty
\frac{\left[ \left(\frac{(x-\mu_i)^2}{\sigma_i^2}-1\right) \mathfrak{n}(x|\mu_i,\sigma_i^2)\right]^2}{\sum_{l=1}^k p_l \mathfrak{n}(x|\mu_l,\sigma_l^2)}
d x
$$
 
$$
\mathbb{E}\left[- \frac{\partial^2 log f_X(X)}{\partial \sigma_i \partial \sigma_j}\right]=\frac{p_i p_j}{\sigma_i \sigma_j} \int_{-\infty}^\infty
\frac{ \left(\frac{(x-\mu_i)^2}{\sigma_i^2}-1\right) \left(\frac{(x-\mu_j)^2}{\sigma_j^2}-1\right) \mathfrak{n}(x|\mu_i,\sigma_i^2)\mathfrak{n}(x|\mu_j,\sigma_j^2)}{\sum_{l=1}^k p_l \mathfrak{n}(x|\mu_l,\sigma_l^2)}
d x
$$

$$
\mathbb{E}\left[- \frac{\partial^2 log f_X(X)}{\partial \mu_i \partial \sigma_i}\right]=\frac{p_i^2}{\sigma_i^3} \int_{-\infty}^\infty
\frac{ \left(x-\mu_i\right)\left(\frac{(x-\mu_i)^2}{\sigma_i^2}-1\right) \left[\mathfrak{n}(x|\mu_i,\sigma_i^2)\right]^2}{\sum_{l=1}^k p_l \mathfrak{n}(x|\mu_l,\sigma_l^2)}
d x
$$
 
$$
\mathbb{E}\left[- \frac{\partial^2 log f_X(X)}{\partial \mu_i \partial \sigma_j}\right]=\frac{p_i p_j}{\sigma_i \sigma_j} \int_{-\infty}^\infty
\frac{ \frac{(x-\mu_i)}{\sigma_i^2 \sigma_j} \left(\frac{(x-\mu_j)^2}{\sigma_j^2}-1\right) \mathfrak{n}(x|\mu_i,\sigma_i^2)\mathfrak{n}(x|\mu_j,\sigma_j^2)}{\sum_{l=1}^k p_l \mathfrak{n}(x|\mu_l,\sigma_l^2)}
d x
$$

\end{proof}

\subsubsection*{Location and scale parameters unknown}

When considering all the parameters unknown, the form of the Jeffreys prior may be partly defined by considering the mixture model as a location-scale model, for which a general solution exists; see \cite{robert:2001}. 

\begin{lemma}
When all the parameters of a Gaussian mixture model are unknown, the Jeffreys prior is constant in $\mu$ and powered as $\tau^{-d/2}$, where $d$ is the total number of components parameters. 
\end{lemma}

\begin{proof}
We have already proved the Jeffreys prior is constant on the global mean (first proof of Lemma \ref{lem:meansd-prior}).

Consider a two-component mixture model and the reparametrization \eqref{reparMix}. With some computations, it is straightforward to derive the Fisher information matrix for this model, partly shown in Table \ref{tab:FishInfo_repar}, where each term is multiplied for a term which does not depend on $\tau$.

\begin{table}[h]
\centering
\caption{Factors depending on $\tau$ of the Fisher information matrix for the reparametrized model \eqref{reparMix}}.
\label{tab:FishInfo_repar}
\begin{tabular}{c|ccccc}
                  & \textbf{$\sigma$} & \textbf{$\delta$} & \textbf{p}                       & \textbf{$\mu$} & \textbf{$\tau$} \\ \hline
\textbf{$\sigma$} & 1                 & 1                 & \multicolumn{1}{c|}{1}           & $\tau^{-1}$      & $\tau^{-1}$     \\
\textbf{$\delta$} & 1                 & 1                 & \multicolumn{1}{c|}{1}           & $\tau^{-1}$      & $\tau^{-1}$     \\
\textbf{p}        & 1                 & 1                 & \multicolumn{1}{c|}{1}           & $\tau^{-1}$      & $\tau^{-1}$     \\ \cline{2-6} 
\textbf{$\mu$}    & $\tau^{-1}$       & $\tau^{-1}$       & \multicolumn{1}{c|}{$\tau^{-1}$} & $\tau^{-2}$      & $\tau^{-2}$     \\
\textbf{$\tau$}   & $\tau^{-1}$       & $\tau^{-1}$       & \multicolumn{1}{c|}{$\tau^{-1}$} & $\tau^{-2}$      & $\tau^{-2}$    
\end{tabular}
\end{table}

Therefore, the Fisher information matrix considered as a function of $\tau$ is a block matrix. From well-known results
in linear algebra, if we consider a block matrix 

\begin{equation*}
M=
\begin{bmatrix}
A & B \\
C & D
\end{bmatrix}
\end{equation*}

\noindent then its determinant is given by $\det(M)=\det(A-BD^{-1}C)\det(D)$. In the case of a two-component mixture
model, $\det(D)\propto\tau^{-4}$, while $\det(A-BD^{-1}C)\propto 1$ (always seen as functions of $\tau$ only). Then the
Jeffreys prior for a two-component location-scale mixture model is proportional to $\tau^{-2}$.

This result may be easily generalized to the case of $k$ components. 

\end{proof}

\subsection{Posterior distributions of Jeffreys priors}

We now derive analytical and computational characterizations of the posterior distributions associated with Jeffreys
priors for mixture models. Simulated examples are used to support the analytical results.

For this purpose, we have repeated simulations from the models 
\begin{equation}
0.50\mathcal{N}(\mu_1,1) + 0.50\mathcal{N}(\mu_2,0.5)
\label{eq:2mix}
\end{equation}
\noindent and
\begin{equation}
0.25\mathcal{N}(\mu_1,1) + 0.65\mathcal{N}(\mu_0,0.5) + 0.10\mathcal{N}(\mu_2,5)
\label{eq:3mix}
\end{equation}
\noindent where $\mu_1$ and $\mu_2$ are chosen to be either close ($\mu_1=-1$, $\mu_2=2$) or well separated ($\mu_1=-10$, $\mu_2=15$) and $\mu_0=0$. 

The Tables shown in the following will analyze the behavior of simulated Markov chains with the goal to approximate the posterior distribution. Even if the output of a MCMC method is not conclusive to assess the properness of the target distribution, it may give a hint on improperness: if the target is improper, an MCMC chain cannot be positive recurrent but instead either null-recurrent or transient \citep{robert:casella:2004}, then it should show convergence problems, as trends or difficulties to move from a particular region. Therefore, simulation studies will be used to support analytical results on properness or improperness of the posterior distribution. 
In the following, we will say that the results are stable if they show a convergent behavior, i.e. they move around the true values which have generated the data. In particular, an approximation is stable if the proportion of experiments for which the chains show no trend and acceptance rates around the expected values (20\%-40\%, which means that there are not regions where the chain have difficulties to move from) is 0. 

The following results are based on Gaussian mixture models, anyway, the Jeffreys prior has a behavior common to all the location-scale families, as shown in Section \ref{subsec:priors}, as well as the likelihood function; therefore the results may be generalized to any location-scale family.

\subsubsection{Location parameters unknown}

A first numerical study where the Jeffreys prior and its posterior are computed on a grid of parameter values confirms
that, provided the means only are unknown, the prior is constant on the difference between the means and takes higher and higher
values as the difference between them increases. However, the posterior distribution is correctly concentrated around the
true values for a sufficiently high sample size and it exhibits the classical bimodal nature of such posteriors \citep{celeux:hurn:robert:2000}.
In Figure \ref{fig:mean-priorpost}, the posterior distribution appears to be perfectly symmetric because the other parameters
(weights and standard deviations) have been fixed as identical.
 
\begin{figure}
\centering
\includegraphics[width=6.5cm, height=7.5cm]{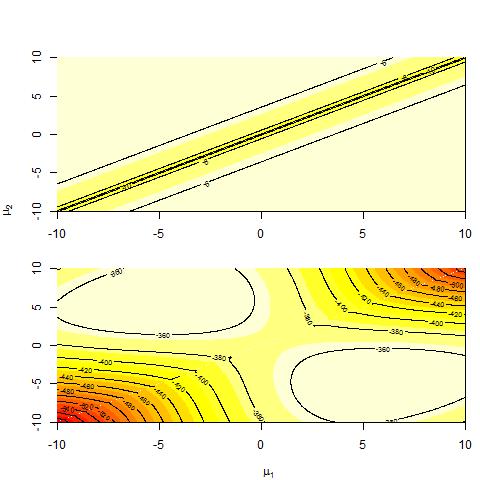}
\caption{Approximations (on a grid of values) of the Jeffreys prior (on the log-scale) when only the means of a Gaussian
mixture model with two components are unknown (on the top) and of the derived posterior distribution (with known weights both equal to 0.5 and known standard deviations both equal to 5).}
\label{fig:mean-priorpost}
\end{figure}

Tables \ref{tab:post2means} and \ref{tab:post3means} show that, when considering a two-component Gaussian mixture model, the results 
are stabilizing for a sample size equal to $10$ if the components are close and they are always stable if the means are far enough; 
on the other hand, huge sample sizes (around $100$ observations) are needed to have always converging chains for a three-component 
mixture model (even if, when the components are well-separated a sample size equal to $10$ seems to be enough to have stable results). 

\begin{lemma} 
When $k=2$, the posterior distribution derived from the Jeffreys prior when only the means are unknown is proper.
\label{lem:mean-post}
\end{lemma}

\begin{proof}
The conditional Jeffreys prior for the means of a Gaussian mixture model is
\begin{align*}
\pi^J(\mu|p,\sigma) &\propto \frac{p_1 p_2}{\sigma_1^2 \sigma_2^2}\left\{ \int_{-\infty}^{+\infty}\frac{\left[
t\mathfrak{n}(0,\sigma_1)\right]^2}{p_1\mathfrak{n}(0,\sigma_1)+p_2\mathfrak{n}(\delta,\sigma_2)} d t \right. \nonumber \\ 
					&{} \left. \times  \int_{-\infty}^{+\infty} \frac{\left[
u\mathfrak{n}(0,\sigma_2)\right]^2}{p_1\mathfrak{n}(-\delta,\sigma_1)+p_2\mathfrak{n}(0,\sigma_2)} d u \right. \nonumber \\ 
					&{} \left. -\left(\int_{-\infty}^{+\infty} \frac{ t\mathfrak{n}(0,\sigma_1)
(t-\delta)\mathfrak{n}(\delta,\sigma_2)}{p_1\mathfrak{n}(0,\sigma_1)+p_2\mathfrak{n}(\delta,\sigma_2)} d t\right)^2 \right\}^\frac{1}{2}
\end{align*}
\noindent where $\delta=\mu_2-\mu_1$. 

The posterior distribution is then defined as
\begin{equation*}
\prod_{j=1}^n \left[p_1\mathfrak{n}(\mu_1,\sigma_1)+p_2\mathfrak{n}(\mu_2,\sigma_2)\right]\pi^J(\mu_1,\mu_2|p,\sigma)
\end{equation*}

The likelihood may be rewritten (without loss of generality, by considering $\sigma_1=\sigma_2=1$, since they are known) as
\begin{align}
L(\theta)&=\prod_{j=1}^n \left[p_1\mathfrak{n}(\mu_1,1)+p_2\mathfrak{n}(\mu_2,1)\right] \nonumber \\
		&= \frac{1}{(2\pi)^\frac{n}{2}}\left[p_1^n e^{-\frac{1}{2}\sum_{i=1}^n (x_i-\mu_1)^2}+\sum_{j=1}^n p_1^{n-1}p_2e^{-\frac{1}{2}\sum_{i\neq j} (x_i-\mu_1)^2-\frac{1}{2}(x_j-\mu_2)^2}\right. \nonumber \\
		&{} \left. +\sum_{j=1}^n \sum_{k\neq j} p_1^{n-2}p_2^2 e^{-\frac{1}{2}\sum_{i\neq j,k} (x_i-\mu_1)^2-\frac{1}{2}\left[(x_j-\mu_2)^2+(x_k-\mu_2)^2 \right]} \right. \nonumber \\
		&{} \left. +\cdots+p_2^n e^{-\frac{1}{2}\sum_{j=1}^n (x_j-\mu_2)^2}\right]
\label{eq:mixlik}
\end{align}
Then, for $|\mu_1|\rightarrow\infty$, $L(\theta)$ tends to the term $p_2^n e^{-\frac{1}{2}\sum_{j=1}^n (x_j-\mu_2)^2}$
that is constant for $\mu_1$. Therefore we can study the behavior of the posterior distribution for this part of the
likelihood to assess its properness. 

This explains why we want the following integral to converge:
\begin{equation*}
\int_{\mathbb{R}\times\mathbb{R}} p_2^n e^{-\frac{1}{2}\sum_{j=1}^n (x_j-\mu_2)^2} \pi^J(\mu_1,\mu_2) d\mu_1 d\mu_2
\end{equation*}
which is equal to (by the change of variable $\mu_2-\mu_1=\delta$)
\begin{equation*}
\int_{\mathbb{R}\times\mathbb{R}} p_2^n e^{-\frac{1}{2}\sum_{j=1}^n (x_j-\mu_1-\delta)^2} \pi^J(\mu_1,\delta) d\mu_1 d\delta
\end{equation*}
We have seen that the prior distribution only depends on the difference between the means $\delta$:
\begin{align}
&\int_\mathbb{R} p_2^n \int_\mathbb{R} e^{-\frac{1}{2}\sum_{j=1}^n (x_j-\mu_1-\delta)^2}d\mu_1 \pi^J(\delta)d\delta \nonumber \\
&\propto \int_\mathbb{R} \int_\mathbb{R} e^{-\frac{1}{2}\sum_{j=1}^n (x_j-\delta)^2
+\mu_1\sum_{j=1}^n(x_j-\delta)-\frac{1}{2}n\mu_1^2} d\mu_1 \pi^J(\delta)d\delta \nonumber \\
&=\int_\mathbb{R} \left[\int_\mathbb{R} e^{\mu_1\sum_{j=1}^n(x_j-\delta)-\frac{1}{2}n\mu_1^2} d\mu_1\right]
e^{-\frac{1}{2}\sum_{j=1}^n (x_j-\delta)^2} \pi^J(\delta)d\delta \nonumber \\
&=\int_\mathbb{R} e^{-\frac{1}{2}\sum_{j=1}^n (x_j-\delta)^2+\sum_{j=1}^n\frac{(x_j-\delta)}{2n}} \pi^J(\delta)d\delta \nonumber \\
&\approx \int_\mathbb{R} e^{-\frac{n-1}{2}\delta^2} \pi^J(\delta)d\delta \label{eq:postmean}
\end{align}

The prior on $\delta$ depends on the determinant of the corresponding Fisher information matrix that is positive definite, then it is bounded by the product of the Fisher information matrix diagonal entries:
\begin{equation}
\footnotesize{
\pi(\delta)\leq \frac{p_1 p_2}{\sigma_1 \sigma_2}\left\{ \bigintsss_{-\infty}^{+\infty}
\frac{\left[t\mathfrak{n}(0,\sigma_1^2)\right]^2}{p_1 \mathfrak{n}(0,\sigma_1^2)+p_2\mathfrak{n}(\delta,\sigma_2^2)} d t
\times \bigintsss_{-\infty}^{+\infty} \frac{\left[u\mathfrak{n}(0,\sigma_2^2)\right]^2}{p_1
\mathfrak{n}(-\delta,\sigma_1^2)+p_2\mathfrak{n}(0,\sigma_2^2)} d u \right\}^\frac{1}{2} 
}
\label{eq:deltaprior}
\end{equation}

\noindent where we have used the proof of lemma \ref{lem:meansd-prior} and a change of variable $(t-\delta)=u$ in the
second integral. As $\delta \rightarrow \pm \infty$ this quantity is constant with respect to $\delta$. Therefore the integral \eqref{eq:postmean} is convergent for $n \geq 2$.
   
\end{proof} 

Unfortunately this result can not be extended to the general case of $k$ components. 

\begin{lemma} 
When $k>2$, the posterior distribution derived from the Jeffreys prior when only the means are unknown is improper.
\label{lem:meankcomp-post}
\end{lemma}

\begin{proof}
In the case of $k\neq 2$ components, the Jeffreys prior for the location parameters is still constant with respect to a
reference mean (for example, $\mu_1$).  Therefore it depends on the difference parameters
$(\delta_2=\mu_2-\mu_1,\delta_3=\mu_3-\mu_1,\cdots,\delta_k=\mu_k-\mu_1)$.

The Jeffreys prior will be bounded by the product on the diagonal, which is an extension of Equation \eqref{eq:deltaprior}:

\begin{align*}
\pi^J(\delta_2,\cdots,\delta_k) &\leq c \left\{ \bigintsss_{-\infty}^\infty \frac{[t\mathfrak{n}(0,\sigma_1^2)]^2}{p_1\mathfrak{n}(0,\sigma_1^2)+\cdots+p_k \mathfrak{n}(\delta_k,\sigma_k^2)}d t \right. \nonumber \\
							& \left. {} \cdots \bigintsss_{-\infty}^\infty \frac{[u \mathfrak{n}(0,\sigma_k^2)]^2}{p_1\mathfrak{n}(-\delta_k,\sigma_1^2)+\cdots+p_k \mathfrak{n}(0,\sigma_k^2)} d u \right\}^\frac{1}{2}.
\end{align*}

If we consider the case as in Lemma \ref{lem:mean-post}, where only the part of the likelihood depending on e.g. $\mu_2$
may be considered, the convergence of the following integral has to be studied:

\begin{equation*}
\int_\mathbb{R} \cdots \int_\mathbb{R} e^{-\frac{n-1}{2}\delta_2^2} \pi^J(\delta_2,\cdots,\delta_k) d \delta_2 \cdots d \delta_k
\end{equation*}

In this case, however, the integral with respect to $\delta_2$ may converge, nevertheless the integrals with respect to $\delta_j$ with $j\neq 2$ will diverge, since the prior tends to be constant for each $\delta_j$ as $|\delta_j| \rightarrow \infty$.

This results confirms the idea that each part of the likelihood  gives information about at most the difference between the location of the respective components and the reference locations, but not on the locations of the other components. 

\end{proof}

\begin{table}[h]
\centering
\footnotesize
\caption{$\mu$ unknown, k=2: results of 50 replications of the experiment for both close and far means with a Monte Carlo approximation of the posterior distribution based on $10^5$ simulations and a burn-in of $10^4$ simulations. The table shows the average acceptance rate, the proportion of chains diverging towards higher values and the average ratio between the log-likelihood of the last accepted values and the true values in the 50 replications when using the Jeffreys prior (on the left) and a prior constant on the means (on the right).}
\label{tab:post2means}
\begin{tabular}{|cccc|ccc|}
\hline
\multicolumn{1}{|l|}{{\bf k=2}}                                                    & \multicolumn{1}{l}{{\bf Jeffreys prior}}                             & \multicolumn{1}{l}{(Close Means)}                                          & \multicolumn{1}{l|}{}                                                                                & \multicolumn{1}{l}{{\bf Constant prior}}                             & \multicolumn{1}{l}{}                                                       & \multicolumn{1}{l|}{}                                                                                \\ \hline
\multicolumn{1}{|c|}{{\it \begin{tabular}[c]{@{}c@{}}Sample \\ Size\end{tabular}}} & {\it \begin{tabular}[c]{@{}c@{}}Ave. \\ Accept.\\ Rate\end{tabular}} & {\it \begin{tabular}[c]{@{}c@{}}Chains towards\\ high values\end{tabular}} & {\it \begin{tabular}[c]{@{}c@{}}Ave. \\ lik($\theta^{fin}$)\\ /\\ lik($\theta^{true}$)\end{tabular}} & {\it \begin{tabular}[c]{@{}c@{}}Ave. \\ Accept.\\ Rate\end{tabular}} & {\it \begin{tabular}[c]{@{}c@{}}Chains towards\\ high values\end{tabular}} & {\it \begin{tabular}[c]{@{}c@{}}Ave. \\ lik($\theta^{fin}$)\\ /\\ lik($\theta^{true}$)\end{tabular}} \\ \hline
\multicolumn{1}{|c|}{2}                                                            & 0.2505                                                               & 0.88                                                                       & 1.8182                                                                                               & 0.2709                                                               & 0.72                                                                       &  1.9968                                                                                                 \\
\multicolumn{1}{|c|}{3}                                                            & 0.2656                                                               & 0.94                                                                       & 1.6804                                                                                               & 0.2782                                                               & 0.58                                                                       &  1.9613                                                                                                  \\
\multicolumn{1}{|c|}{4}                                                            & 0.2986                                                               & 0.56                                                                       & 1.3097                                                                                               & 0.2812                                                               & 0.18                                                                       & 1.9824                                                                                                  \\
\multicolumn{1}{|c|}{5}                                                            & 0.2879                                                               & 0.48                                                                       & 1.2918                                                                                               & 0.2830                                                               & 0.14                                                                       & 1.8358
\\
\multicolumn{1}{|c|}{6}                                                            & 0.3066                                                               & 0.16                                                                       & 1.1251                                                                                               & 0.3090                                                               & 0.00                                                                       & 1.9363                                                                                               \\
\multicolumn{1}{|c|}{7}                                                            & 0.3052                                                               & 0.24                                                                       & 1.1205                                                                                               & 0.3103                                                               & 0.02                                                                       & 1.7994                                                                                               \\
\multicolumn{1}{|c|}{8}                                                            & 0.3181                                                               & 0.02                                                                       & 1.0149                                                                                               & 0.3521                                                               & 0.00                                                                       & 1.3923                                                                                               \\
\multicolumn{1}{|c|}{9}                                                            & 0.3101                                                               & 0.02                                                                       & 1.0244                                                                                               & 0.3369                                                               & 0.00                                                                       & 1.5219                                                                                               \\
\multicolumn{1}{|c|}{10}                                                           & 0.3460                                                               & 0.00                                                                       & 0.9914                                                                                               & 0.3627                                                               & 0.00                                                                       & 1.2933                                                                                               \\
\multicolumn{1}{|c|}{15}                                                           & 0.3418                                                               & 0.00                                                                       & 1.0097                                                                                               & 0.3913                                                               & 0.00                                                                       & 1.1970                                                                                               \\
\multicolumn{1}{|c|}{20}                                                           & 0.3881                                                               & 0.00                                                                       & 0.9948                                                                                               & 0.4097                                                               & 0.00                                                                       & 1.1032                                                                                               \\
\multicolumn{1}{|c|}{50}                                                           & 0.4556                                                               & 0.00                                                                       & 1.0005                                                                                               & 0.4515                                                               & 0.00                                                                       & 1.0303                                                                                               \\
\multicolumn{1}{|c|}{100}                                                          & 0.5090                                                               & 0.00                                                                       & 1.0008                                                                                               & 0.5090                                                               & 0.00                                                                       & 1.0007                                                                                               \\
\multicolumn{1}{|c|}{500}                                                          & 0.5603                                                               & 0.00                                                                       & 1.0006                                                                                               & 0.5305                                                               & 0.00                                                                       & 1.0002                                                                                               \\
\multicolumn{1}{|c|}{1000}                                                         & 0.4915                                                               & 0.00                                                                       & 1.0006                                                                                               & 0.2327                                                               & 0.00                                                                       & 1.0042                                                                                               \\ \hline
\multicolumn{1}{|c|}{{\bf k=2}}                                                              & \multicolumn{1}{l}{}                                        & \multicolumn{1}{l}{(Far means)}                                            & \multicolumn{1}{l|}{}                                                                                & \multicolumn{1}{l}{{\bf Constant prior}}                             & \multicolumn{1}{l}{}                                                       & \multicolumn{1}{l|}{}                                                                                \\ \hline
\multicolumn{1}{|c|}{2}                                                            & 0.2752                                                               & 0.00                                                                       & 1.0838                                                                                               & 0.2736                                                               & 0.00                                                                       & 1.0474                                                                                               \\
\multicolumn{1}{|c|}{3}                                                            & 0.2692                                                               & 0.00                                                                       & 1.0313                                                                                               & 0.2546                                                               & 0.00                                                                       & 1.0313                                                                                               \\
\multicolumn{1}{|c|}{4}                                                            & 0.2969                                                               & 0.00                                                                       & 1.1385                                                                                               & 0.3152                                                               & 0.00                                                                       & 1.0167                                                                                               \\
\multicolumn{1}{|c|}{5}                                                            & 0.2938                                                               & 0.00                                                                       & 1.0138                                                                                               & 0.2920                                                               & 0.00                                                                       & 0.9968                                                                                               \\
\multicolumn{1}{|c|}{6}                                                            & 0.3066                                                               & 0.00                                                                       & 1.2207                                                                                               & 0.3470                                                               & 0.00                                                                       & 0.9975                                                                                               \\
\multicolumn{1}{|c|}{7}                                                            & 0.3350                                                               & 0.00                                                                       & 1.1055                                                                                               & 0.3473                                                               & 0.00                                                                       & 0.9920                                                                                               \\
\multicolumn{1}{|c|}{8}                                                            & 0.3154                                                               & 0.00                                                                       & 1.1374                                                                                               & 0.3583                                                               & 0.00                                                                       & 1.0092                                                                                               \\
\multicolumn{1}{|c|}{9}                                                            & 0.3309                                                               & 0.00                                                                       & 1.1566                                                                                               & 0.3512                                                               & 0.00                                                                       & 0.9893                                                                                               \\
\multicolumn{1}{|c|}{10}                                                           & 0.3338                                                               & 0.00                                                                       & 1.1820                                                                                               & 0.3601                                                               & 0.00                                                                       & 1.0112                                                                                               \\
\multicolumn{1}{|c|}{15}                                                           & 0.3579                                                               & 0.00                                                                       & 1.1796                                                                                               & 0.3840                                                               & 0.00                                                                       & 1.0136                                                                                               \\
\multicolumn{1}{|c|}{20}                                                           & 0.3950                                                               & 0.00                                                                       & 1.1615                                                                                               & 0.4190                                                               & 0.00                                                                       & 1.0096                                                                                               \\
\multicolumn{1}{|c|}{50}                                                           & 0.4879                                                               & 0.00                                                                       & 1.1682                                                                                               & 0.4659                                                               & 0.00                                                                       & 1.0059                                                                                               \\
\multicolumn{1}{|c|}{100}                                                          & 0.5083                                                               & 0.00                                                                       & 1.2123                                                                                               & 0.4957                                                               & 0.00                                                                       & 1.0017                                                                                               \\
\multicolumn{1}{|c|}{500}                                                          & 0.5570                                                               & 0.00                                                                       & 1.1996                                                                                               & 0.4777                                                               & 0.00                                                                            & 0.9976                                                                                               \\
\multicolumn{1}{|c|}{1000}                                                         & 0.3463                                                               & 0.00                                                                       & 1.2161                                                                                               & 0.1792                                                               & 0.00                                                                       & 1.0010                                                                                               \\ \hline
\end{tabular}
\end{table}

\begin{table}[h]
\centering
\small
\caption{$\mu$ unknown, k=3: as in Table \ref{tab:post2means} for two three-component Gaussian mixture models, with close and far means, only for the Jeffreys prior.}
\label{tab:post3means}
\begin{tabular}{|cccc|}
\hline
\multicolumn{1}{|l|}{{\bf k=3}}                                                    & \multicolumn{1}{l}{{\bf Jeffreys prior}}                             & \multicolumn{1}{l}{(Close Means)}                                          & \multicolumn{1}{l|}{}                                                                                \\ \hline
\multicolumn{1}{|c|}{{\it \begin{tabular}[c]{@{}c@{}}Sample \\ Size\end{tabular}}} & {\it \begin{tabular}[c]{@{}c@{}}Ave. \\ Accept.\\ Rate\end{tabular}} & {\it \begin{tabular}[c]{@{}c@{}}Chains towards\\ high values\end{tabular}} & {\it \begin{tabular}[c]{@{}c@{}}Ave. \\ lik($\theta^{fin}$)\\ /\\ lik($\theta^{true}$)\end{tabular}} \\ \hline
\multicolumn{1}{|c|}{2}                                                            & 0.2366                                                               & 1.00                                                                       & 2.5175                                                                                               \\
\multicolumn{1}{|c|}{3}                                                            & 0.2608                                                               & 1.00                                                                       & 2.8447                                                                                               \\
\multicolumn{1}{|c|}{4}                                                            & 0.2455                                                               & 0.98                                                                       & 1.3749                                                                                               \\
\multicolumn{1}{|c|}{5}                                                            & 0.2446                                                               & 1.00                                                                       & 1.3807                                                                                               \\
\multicolumn{1}{|c|}{6}                                                            & 0.2330                                                               & 1.00                                                                       & 1.4062                                                                                               \\
\multicolumn{1}{|c|}{7}                                                            & 0.2480                                                               & 0.98                                                                       & 1.2411                                                                                               \\
\multicolumn{1}{|c|}{8}                                                            & 0.2684                                                               & 0.94                                                                       & 1.2535                                                                                               \\
\multicolumn{1}{|c|}{9}                                                            & 0.2784                                                               & 0.98                                                                       & 1.2744                                                                                               \\
\multicolumn{1}{|c|}{10}                                                           & 0.2904                                                               & 0.68                                                                       & 1.1168                                                                                               \\
\multicolumn{1}{|c|}{15}                                                           & 0.3214                                                               & 0.74                                                                       & 1.1217                                                                                               \\
\multicolumn{1}{|c|}{20}                                                           & 0.3819                                                               & 0.32                                                                       & 1.0616                                                                                               \\
\multicolumn{1}{|c|}{30}                                                           & 0.3774                                                               & 0.10                                                                       & 1.0383                                                                                               \\
\multicolumn{1}{|c|}{50}                                                           & 0.4407                                                               & 0.04                                                                       & 1.0108                                                                                               \\
\multicolumn{1}{|c|}{100}                                                          & 0.4935                                                               & 0.00                                                                       & 1.0018                                                                                               \\
\multicolumn{1}{|c|}{500}                                                          & 0.5577                                                               & 0.00                                                                       & 1.0068                                                                                               \\
\multicolumn{1}{|c|}{1000}                                                         & 0.5511                                                               & 0.00                                                                       & 1.0006                                                                                               \\ \hline
\multicolumn{1}{|c|}{{\bf k=3}}                                                              & \multicolumn{1}{l}{}                                        & \multicolumn{1}{l}{(Far means)}                                            & \multicolumn{1}{l|}{}                                                                                \\ \hline
\multicolumn{1}{|c|}{2}                                                            & 0.2641                                                               & 1.00                                                                       & 2.1786                                                                                               \\
\multicolumn{1}{|c|}{3}                                                            & 0.2804                                                               & 1.00                                                                       & 2.1039                                                                                               \\
\multicolumn{1}{|c|}{4}                                                            & 0.2813                                                               & 0.82                                                                       & 1.1173                                                                                               \\
\multicolumn{1}{|c|}{5}                                                            & 0.2840                                                               & 0.84                                                                       & 1.0412                                                                                               \\
\multicolumn{1}{|c|}{6}                                                            & 0.2887                                                               & 0.84                                                                       & 1.1050                                                                                               \\
\multicolumn{1}{|c|}{7}                                                            & 0.2865                                                               & 0.82                                                                       & 1.0840                                                                                               \\
\multicolumn{1}{|c|}{8}                                                            & 0.3248                                                               & 0.66                                                                       & 1.0982                                                                                               \\
\multicolumn{1}{|c|}{9}                                                            & 0.3277                                                               & 0.76                                                                       & 1.1177                                                                                               \\
\multicolumn{1}{|c|}{10}                                                           & 0.2998                                                               & 0.00                                                                       & 1.2604                                                                                               \\
\multicolumn{1}{|c|}{15}                                                           & 0.3038                                                               & 0.00                                                                       & 1.3149                                                                                               \\
\multicolumn{1}{|c|}{20}                                                           & 0.2869                                                               & 0.00                                                                       & 1.3533                                                                                               \\
\multicolumn{1}{|c|}{30}                                                           & 0.3762                                                               & 0.00                                                                       & 1.2479                                                                                               \\
\multicolumn{1}{|c|}{50}                                                           & 0.4283                                                               & 0.00                                                                       & 1.3791                                                                                               \\
\multicolumn{1}{|c|}{100}                                                          & 0.5251                                                               & 0.00                                                                       & 1.2585                                                                                               \\
\multicolumn{1}{|c|}{500}                                                          & 0.5762                                                               & 0.00                                                                       & 1.4779                                                                                               \\
\multicolumn{1}{|c|}{1000}                                                         & 0.4751                                                               & 0.00                                                                       & 1.2161                                                                                               \\ \hline
\end{tabular}
\end{table}

\subsubsection{Scale parameters unknown}

\begin{lemma} 
The posterior distribution derived from the Jeffreys prior when only the standard deviations are unknown is improper.
\label{lem:sd-post}
\end{lemma}

\begin{proof}
Consider equation \eqref{eq:mixlik} generalized to the case of $\sigma_1$ and $\sigma_2$ unknown: then when we integrate
the posterior distribution with respect to $\sigma_1$ and $\sigma_2$, the complete integral may be split into several
integrals then summed up. In particular, if we consider the first part of the likelihood (which only depends on the
first component of the mixture) and use the change of variable used in \eqref{eq:scaleprior}, we have:

\begin{align*}
\bigintsss_0^{\infty} \bigintsss_0^{\infty} & c \frac{p_1^n}{\tau^n}\frac{p_1 p_2}{\tau\sigma} \exp \left\{-\frac{1}{2\tau^2}\sum_{i=1}^n(x_i-\mu_1)^2\right\} \nonumber \\
				& \times{} \left\{ \bigintsss_{-\infty}^\infty
\frac{\left(z^2-1\right)^2 \exp\left\{ -z^2\right\}}{p_1 \exp\left\{-\frac{z^2}{2}\right\}+\frac{p_2}{\sigma}\exp\left\{-\frac{(z\tau+\mu_1-\mu_2)^2}{2\tau^2\sigma^2}\right\}}
d z \right. \nonumber \\
					&\times \left.{} \bigintsss_{-\infty}^\infty
\frac{\left(u^2-1\right)^2 \exp\left\{-u^2\right\}}{p_1\sigma \exp\left\{-\frac{(u\tau\sigma+\mu_2-\mu_1)^2}{2\tau^2}\right\}+p_2\exp\left\{-\frac{u^2}{2}\right\}} 
d u \right\}^\frac{1}{2} d \tau d \sigma
\end{align*}

The integral with respect to $\tau$ in the previous equation converges, nevertheless the likelihood does not provide information for $\sigma$, then the integral with respect to $\sigma$ diverges and the posterior will be improper.

This results may be easily extented to the case of $k$ components: there is a part of the likelihood which only depends on the global scale parameter and is not informative for the ay other components; the form of the integral will remain the same, with integrations with respect to $\sigma_1,\sigma_2,\cdots,\sigma_k$ which do not converge.

\end{proof}

When only the standard deviations are unknown, the Jeffreys prior is concentrated around $0$.
Nevertheless, the posterior distribution shown in Figures \ref{fig:sd-priorpost-farm} turns out to be concentrated around the true values of the parameters for a sufficient high sample size (in the figures, $n$ is always equal to $100$). 

\begin{figure}
\centering
\includegraphics[width=6.5cm, height=7.5cm]{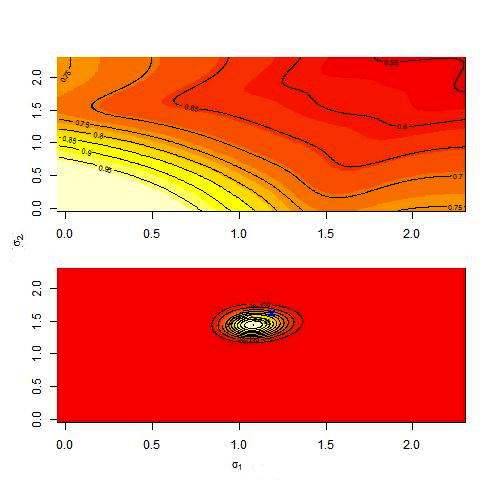}
\caption{Approximations (on a grid of values) of the Jeffreys prior (on the log-scale) when only the standard deviations of a Gaussian mixture model with 2 components are unknown (on the top) and of the derived posterior distribution (with known weights both equal to $0.5$ and known means equal to $(-5,5)$). The blue cross represents the maximum likelihood estimates.}
\label{fig:sd-priorpost-farm}
\end{figure}

Figures \ref{fig:sd-priorpost-clm} and \ref{fig:sd-priorpost-asym} show the prior and the posterior distributions of the scale parameters of a two-component mixture model for some situations with different weights and different means.

\begin{figure}
\centering
\includegraphics[width=6.5cm, height=7.5cm]{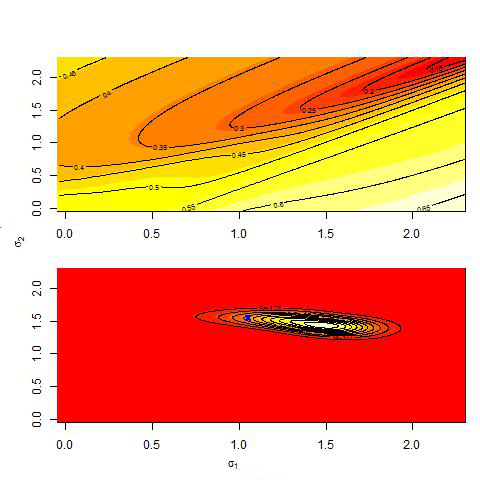}
\caption{Same as Figure \ref{fig:sd-priorpost-farm} but with known weights equal to $(0.25,0.75)$ and known means equal to $(-1,1)$.}
\label{fig:sd-priorpost-clm}
\end{figure}

\begin{figure}
\centering
\includegraphics[width=6.5cm, height=7.5cm]{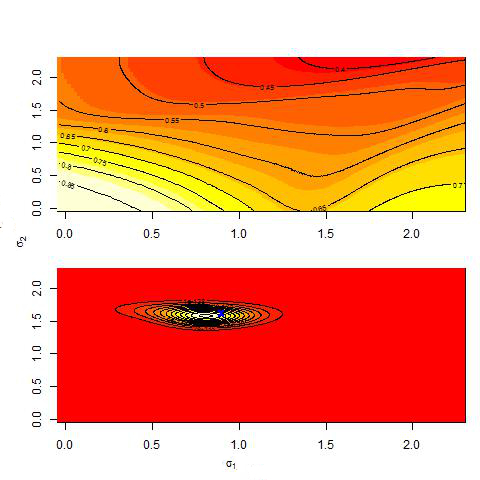}
\caption{Same as Figure \ref{fig:sd-priorpost-farm} but with known weights equal to $(0.25,0.75)$ and known means equal to $(-2,7)$.}
\label{fig:sd-priorpost-asym}
\end{figure}

Summarized results of the posterior approximation obtained via a random-walk Metropolis-Hastings algorithm by exploring
the posterior distribution associated with the Jeffreys prior on the standard deviations are shown in Figures
\ref{fig:sd2-bxp} and \ref{fig:sd3-bxp}, which display boxplots of the posterior means: provided a sufficiently high
sample size, simulations exhibit a convergent behavior. 

\begin{figure}
\centering
\includegraphics[width=6.5cm, height=7.5cm]{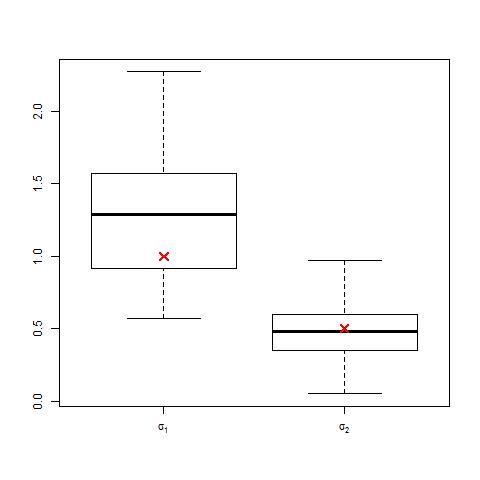}
\caption{Boxplots of posterior means of the standard deviations of the two-component mixture model $0.50\mathcal{N}(-1,1) + 0.50\mathcal{N}(2,0.5)$ for 50 replications of the experiment and a sample size equal to $10$, obtained via MCMC with $10^5$ simulations. The red cross represents the true values.}
\label{fig:sd2-bxp}
\end{figure} 

\begin{figure}
\centering
\includegraphics[width=6.5cm, height=7.5cm]{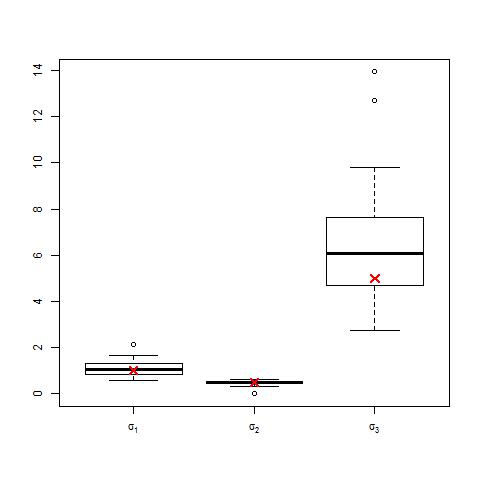}
\caption{Boxplots of posterior means of the standard deviations of the three-component mixture model $0.25\mathcal{N}(-1,1) + 0.65\mathcal{N}(0,0.5) + 0.10\mathcal{N}(2,5)$ for 50 replications of the experiment and a sample size equal to $50$, obtained via MCMC with $10^5$ simulations. The red cross represents the true values.}
\label{fig:sd3-bxp}
\end{figure} 

Repeated simulations show that, for a Gaussian mixture model with two components, a sample size equal to $10$ is necessary to have convergent results, while for a three-component Gaussian mixture model with a sample size equal to $50$ is still possible to have chains stuck to values of standard deviations close to $0$. 

Table \ref{tab:post2sd} and \ref{tab:post3sd} show results for repeated simulations in the cases of two-component and three-component Gaussian mixture models with unknown standard deviations, respectively, where the means that generate the data may be close or far from one another. In Table \ref{tab:post2sd} it seems that the chains tend to be convergent for sample sizes smaller than $10$, but in Table \ref{tab:post3sd} one may see that even with a high sample size (equal to $50$) it may happens, for $k=3$, that the chains are stuck to very small values of standard deviations and this fact confirms what we have proved with Lemma \ref{lem:sd-post}.

\begin{table}[h]
\centering
\footnotesize
\caption{$\sigma$ unknown, $k=2$: results of 50 replications of the experiment for both close and far means with a Monte Carlo approximation of the posterior distribution based on $10^5$ simulations and a burn-in of $10^4$ simulations. The table shows the average acceptance rate, the proportion of chains stuck at values of standard deviations close to 0 and the average ratio between the log-likelihood of the last accepted values and the true values in the 50 replications when using the Jeffreys prior.}
\label{tab:post2sd}
\begin{tabular}{|cccc|}
\hline
\multicolumn{1}{|l|}{{\bf k=2}}                                                    & \multicolumn{1}{l}{{\bf Jeffreys prior}}                             & \multicolumn{1}{l}{(Close Means)}                                                             & \multicolumn{1}{l|}{}                                                                                \\ \hline
\multicolumn{1}{|c|}{{\it \begin{tabular}[c]{@{}c@{}}Sample \\ Size\end{tabular}}} & {\it \begin{tabular}[c]{@{}c@{}}Ave. \\ Accept.\\ Rate\end{tabular}} & {\it \begin{tabular}[c]{@{}c@{}}Chains stuck\\ at small values\\ of $\sigma$\end{tabular}} & {\it \begin{tabular}[c]{@{}c@{}}Ave. \\ lik($\theta^{fin}$)\\ /\\ lik($\theta^{true}$)\end{tabular}} \\ \hline
\multicolumn{1}{|c|}{2}                                                            & 0.2414                                                               & 0.02                                                                                          & 1.2245                                                                                               \\
\multicolumn{1}{|c|}{3}                                                            & 0.1875                                                              & 0.02                                                                                          & 1.1976                                                                                               \\
\multicolumn{1}{|c|}{4}                                                            & 0.2403                                                               & 0.00                                                                                          & 1.0720                                                                                               \\
\multicolumn{1}{|c|}{5}                                                            & 0.2233                                                               & 0.02                                                                                          & 1.1269                                                                                               \\
\multicolumn{1}{|c|}{6}                                                            & 0.2475                                                               & 0.00                                                                                          & 1.0553                                                                                               \\
\multicolumn{1}{|c|}{7}                                                            & 0.2494                                                               & 0.02                                                                                          & 1.0324                                                                                               \\
\multicolumn{1}{|c|}{8}                                                            & 0.2465                                                               & 0.00                                                                                          & 1.0093                                                                                               \\
\multicolumn{1}{|c|}{9}                                                            & 0.2449                                                               & 0.00                                                                                          & 1.0026                                                                                               \\
\multicolumn{1}{|c|}{10}                                                           & 0.2476                                                               & 0.00                                                                                          & 0.9960                                                                                               \\
\multicolumn{1}{|c|}{15}                                                           & 0.2541                                                               & 0.00                                                                                          & 0.9959                                                                                               \\
\multicolumn{1}{|c|}{20}                                                           & 0.2480                                                               & 0.00                                                                                          & 0.9946                                                                                               \\
\multicolumn{1}{|c|}{30}                                                           & 0.2364                                                               & 0.00                                                                                          & 1.0052                                                                                               \\
\multicolumn{1}{|c|}{50}                                                           & 0.2510                                                               & 0.00                                                                                          & 0.9981                                                                                               \\
\multicolumn{1}{|c|}{100}                                                          & 0.3033                                                               & 0.00                                                                                          & 0.9994                                                                                               \\
\multicolumn{1}{|c|}{500}                                                          & 0.4314                                                               & 0.00                                                                                          & 0.9999                                                                                               \\
\multicolumn{1}{|c|}{1000}                                                         & 0.4353                                                               & 0.00                                                                                          & 1.0001                                                                                               \\ \hline
{\bf k=2}                                                                          & \multicolumn{1}{|c}{{\bf }}                                           & \multicolumn{1}{l}{(Far means)}                                                               & \multicolumn{1}{l}{}                                                                                \\ \hline
\multicolumn{1}{|c|}{2}                                                            & 0.2262                                                               & 0.14                                                                                          & 1.09202                                                                                              \\
\multicolumn{1}{|c|}{3}                                                            & 0.2384                                                               & 0.10                                                                                          & 1.0536                                                                                               \\
\multicolumn{1}{|c|}{4}                                                            & 0.2542                                                               & 0.02                                                                                          & 1.0281                                                                                               \\
\multicolumn{1}{|c|}{5}                                                            & 0.2502                                                               & 0.04                                                                                          & 0.9932                                                                                               \\
\multicolumn{1}{|c|}{6}                                                            & 0.2550                                                               & 0.00                                                                                          & 0.9981                                                                                               \\
\multicolumn{1}{|c|}{7}                                                            & 0.2554                                                               & 0.00                                                                                          & 0.9569                                                                                               \\
\multicolumn{1}{|c|}{8}                                                            & 0.2473                                                               & 0.00                                                                                          & 0.9929                                                                                               \\
\multicolumn{1}{|c|}{9}                                                            & 0.2481                                                               & 0.00                                                                                          & 0.9888                                                                                               \\
\multicolumn{1}{|c|}{10}                                                           & 0.2402                                                               & 0.00                                                                                          & 0.9969                                                                                               \\
\multicolumn{1}{|c|}{15}                                                           & 0.2431                                                               & 0.00                                                                                          & 0.9988                                                                                               \\
\multicolumn{1}{|c|}{20}                                                           & 0.2416                                                               & 0.00                                                                                          & 0.9998                                                                                               \\
\multicolumn{1}{|c|}{30}                                                           & 0.2453                                                               & 0.04                                                                                          & 1.0016                                                                                               \\
\multicolumn{1}{|c|}{50}                                                           & 0.2550                                                               & 0.00                                                                                          & 0.9992                                                                                               \\
\multicolumn{1}{|c|}{100}                                                          & 0.2359                                                               & 0.00                                                                                          & 0.9999                                                                                               \\
\multicolumn{1}{|c|}{500}                                                          & 0.3000                                                               & 0.00                                                                                          & 1.0001                                                                                               \\
\multicolumn{1}{|c|}{1000}                                                         & 0.3345                                                               & 0.00                                                                                          & 1.0000                                                                                               \\ \hline
\end{tabular}
\end{table}

\begin{table}[h]
\centering
\caption{$\sigma$ unknown, $k=3$: as in table \ref{tab:post2sd} for two three-components Gaussian mixture models, with close and far means.}
\label{tab:post3sd}
\begin{tabular}{|cccc|}
\hline
\multicolumn{1}{|l|}{{\bf k=3}}                                                    & \multicolumn{1}{l}{{\bf Jeffreys prior}}                             & \multicolumn{1}{l}{(Close Means)}                                                             & \multicolumn{1}{l|}{}                                                                                \\ \hline
\multicolumn{1}{|c|}{{\it \begin{tabular}[c]{@{}c@{}}Sample \\ Size\end{tabular}}} & {\it \begin{tabular}[c]{@{}c@{}}Ave. \\ Accept.\\ Rate\end{tabular}} & {\it \begin{tabular}[c]{@{}c@{}}Chains stuck\\ at small values\\ of $\sigma$\end{tabular}} & {\it \begin{tabular}[c]{@{}c@{}}Ave. \\ lik($\theta^{fin}$)\\ /\\ lik($\theta^{true}$)\end{tabular}} \\ \hline
\multicolumn{1}{|c|}{2}                                                            & 0.0441                                                               & 0.88                                                                                          & 0.1206                                                                                               \\
\multicolumn{1}{|c|}{5}                                                            & 0.0659                                                               & 0.72                                                                                          & 1.0638                                                                                               \\
\multicolumn{1}{|c|}{6}                                                            & 0.0621                                                               & 0.70                                                                                          & 1.1061                                                                                               \\
\multicolumn{1}{|c|}{7}                                                            & 0.1013                                                               & 0.54                                                                                          & 1.0655                                                                                               \\
\multicolumn{1}{|c|}{8}                                                            & 0.0781                                                               & 0.52                                                                                          & 1.0880                                                                                               \\
\multicolumn{1}{|c|}{9}                                                            & 0.0729                                                               & 0.60                                                                                          & 1.1003                                                                                               \\
\multicolumn{1}{|c|}{10}                                                           & 0.1506                                                               & 0.26                                                                                          & 1.0516                                                                                               \\
\multicolumn{1}{|c|}{15}                                                           & 0.1689                                                               & 0.18                                                                                          & 1.0493                                                                                               \\
\multicolumn{1}{|c|}{20}                                                           & 0.2322                                                               & 0.10                                                                                          & 1.0478                                                                                               \\
\multicolumn{1}{|c|}{30}                                                           & 0.2366                                                               & 0.00                                                                                          & 1.0125                                                                                               \\
\multicolumn{1}{|c|}{50}                                                           & 0.4407                                                               & 0.02                                                                                          & 1.0061                                                                                               \\
\multicolumn{1}{|c|}{100}                                                          & 0.2666                                                               & 0.00                                                                                          & 1.0021                                                                                               \\
\multicolumn{1}{|c|}{500}                                                          & 0.3871                                                               & 0.00                                                                                          & 1.0003                                                                                               \\
\multicolumn{1}{|c|}{1000}                                                         & 0.4353                                                               & 0.00                                                                                          & 1.0001                                                                                               \\ \hline
\multicolumn{1}{|l}{{\bf k=3}}                                                              & \multicolumn{1}{|c}{}                                        & \multicolumn{1}{l}{(Far means)}                                                               & \multicolumn{1}{l|}{}                                                                                \\ \hline
\multicolumn{1}{|c|}{2}                                                            & 0.0222                                                               & 0.78                                                                                          & 1.0045                                                                                               \\
\multicolumn{1}{|c|}{5}                                                            & 0.0610                                                               & 0.44                                                                                          & 1.0427                                                                                               \\
\multicolumn{1}{|c|}{6}                                                            & 0.0567                                                               & 0.52                                                                                          & 1.0317                                                                                               \\
\multicolumn{1}{|c|}{7}                                                            & 0.0779                                                               & 0.46                                                                                          & 1.0147                                                                                               \\
\multicolumn{1}{|c|}{8}                                                            & 0.0862                                                               & 0.32                                                                                          & 1.0244                                                                                               \\
\multicolumn{1}{|c|}{9}                                                            & 0.1312                                                               & 0.26                                                                                          & 1.0027                                                                                               \\
\multicolumn{1}{|c|}{10}                                                           & 0.1472                                                               & 0.18                                                                                          & 1.0350                                                                                               \\
\multicolumn{1}{|c|}{15}                                                           & 0.15884                                                              & 0.14                                                                                          & 1.0170                                                                                               \\
\multicolumn{1}{|c|}{20}                                                           & 0.2331                                                               & 0.06                                                                                          & 1.0092                                                                                               \\
\multicolumn{1}{|c|}{30}                                                           & 0.2464                                                               & 0.04                                                                                          & 1.0062                                                                                               \\
\multicolumn{1}{|c|}{50}                                                           & 0.2498                                                               & 0.00                                                                                          & 1.0017                                                                                               \\
\multicolumn{1}{|c|}{100}                                                          & 0.2567                                                               & 0.00                                                                                          & 1.0008                                                                                               \\
\multicolumn{1}{|c|}{500}                                                          & 0.2594                                                               & 0.00                                                                                          & 0.9999                                                                                               \\
\multicolumn{1}{|c|}{1000}                                                         & 0.3073                                                               & 0.00                                                                                          & 1.2161                                                                                               \\ \hline
\end{tabular}
\end{table} 

\subsubsection{Location and weight parameters unknown.}

Figure \ref{fig:MW-bxp} shows the boxplots of repeated simulations when both the weights and the means are unknown. 
It is evident that the posterior chains are concentrated around the true values, neverthless some chains (the 14\% of
the replications) show a drift to very high values (in absolute value) and this behavior suggests
improperness of the posterior distribution. 

\begin{figure}
\centering
\includegraphics[width=6.5cm, height=7.5cm]{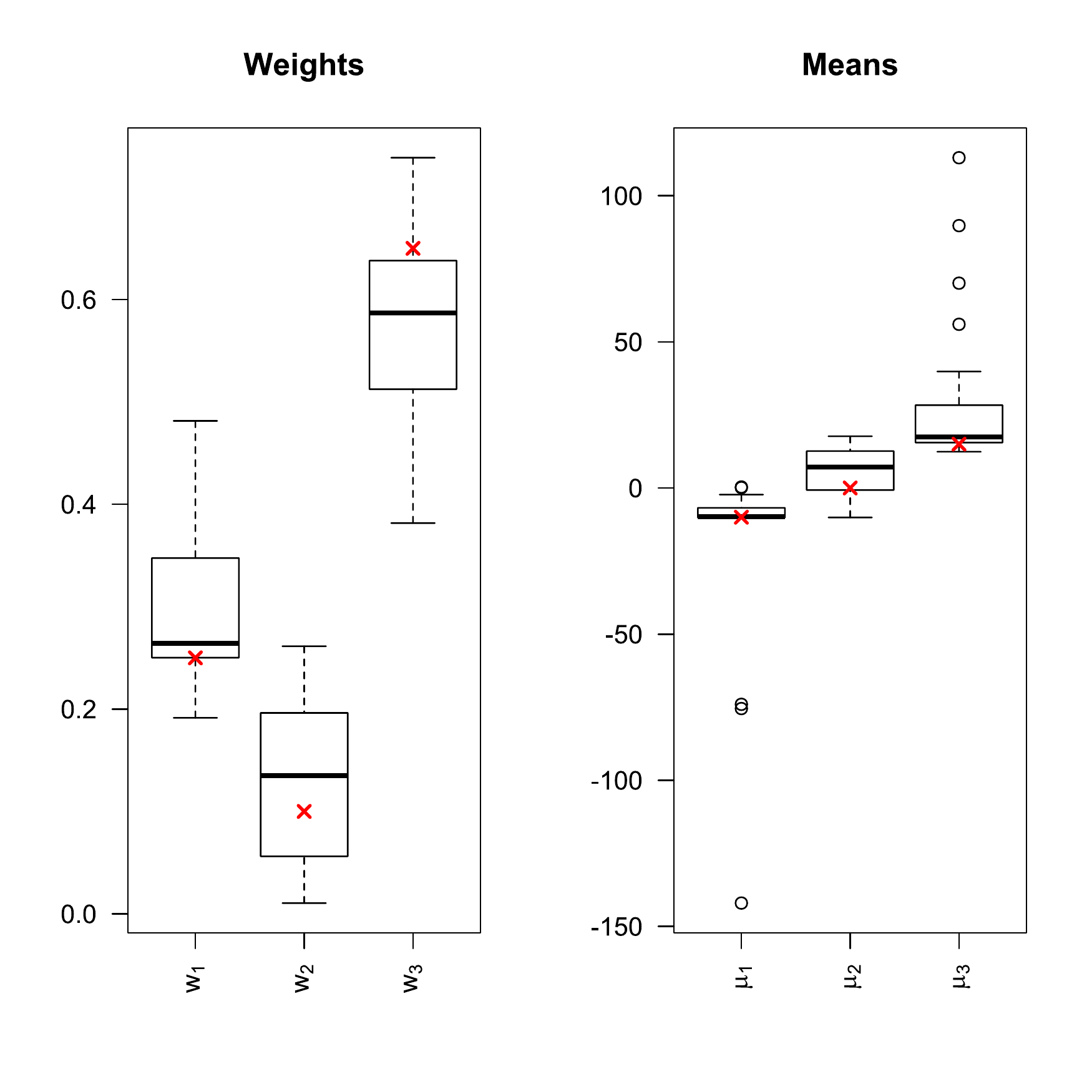}
\caption{Boxplots of posterior means of the weigths and the means of the three-component mixture model $0.25\mathcal{N}(-1,1) + 0.65\mathcal{N}(0,0.5) + 0.10\mathcal{N}(2,5)$ for 50 replications of the experiment, obtained via MCMC with $10^5$ simulations. The red cross represents the true value.}
\label{fig:MW-bxp}
\end{figure} 

\subsubsection{All the parameters unknown}
\label{sub:post}

Improperness of the prior does not imply improperness of the posterior, obviously,
but requires a careful checking of whether or not the posterior is proper, however the proof of Lemma \ref{lem:sd-post} gives an hint about the actual properness of the posterior distribution when all
the parameters are unknown. 

\begin{theorem}
The posterior distribution derived from the Jeffreys prior when all the parameters are unknown is improper.
\label{lem:all-post}
\end{theorem}

\begin{proof}
Consider the elements on the diagonal of the Fisher information matrix; again, since the Fisher information matrix is positive definite, the determinant is bounded by the product of the terms in the diagonal. 

Consider a reparametrization into $\tau=\sigma_1$ and $\tau\sigma=\sigma_2$. Then it is straightforward to see that the
integral of this part of the prior distribution will depend on a term $(\tau)^{-(d+1)}(\sigma)^{-d}$. Again, as in the
proof of Lemma \ref{lem:sd-post}, when composing the prior with the part of the likelihood which only depends on the
first component, this part does not provide information about the parameters $\sigma$ and the integral will diverge. 

In particular, the integral of the first part of the posterior distribution relative to the part of the likelihood dependent on the first component only and on the product of the diagonal terms of the Fisher information matrix for the prior when considering a two-component mixture model is

\begin{align*}
\int_0^1  & \int_{-\infty}^{+\infty} \int_{-\infty}^{+\infty} \int_0^{\infty} \int_0^{\infty} c \frac{p_1^n}{\tau^n}\frac{p_1^2 p_2^2}{\tau^3\sigma^2} \exp \left\{-\frac{1}{2\tau^2}\sum_{i=1}^n(x_i-\mu_1)^2\right\} \nonumber \\
				& \times{} \left\{ \int_{-\infty}^{\infty} \frac{\left[ \sigma\exp\left\{-\frac{(\tau\sigma y + \delta)^2}{2\tau^2} \right\} - \exp\left\{-\frac{y^2}{2} \right\}\right]^2}{p_1 \sigma \exp\left\{-\frac{(\tau \sigma y + \delta)^2}{2\tau^2}\right\}+p_2\exp\left\{-\frac{y^2}{2}\right\}}
d y \right. \nonumber \\
				& \times \left.{} \int_{-\infty}^{\infty} \frac{z^2 \exp(-z^2)}{p_1 \exp\left\{-\frac{z^2}{2}\right\}+\frac{p_2}{\sigma}\exp\left\{-\frac{(z\tau-\delta)^2}{2\tau^2\sigma^2}\right\}}
d z \right. \nonumber \\
				& \times \left.{} \int_{-\infty}^\infty
\frac{w^2 \exp\left\{ -w^2\right\}}{p_1 \sigma \exp\left\{-\frac{(\tau \sigma w+\delta)^2}{2\tau^2\sigma^2}\right\}+p_2\exp\left\{-\frac{w^2}{2}\right\}}
d w \right. \nonumber \\
				& \times \left.{} \int_{-\infty}^\infty
\frac{\left(z^2-1\right)^2 \exp\left\{ -z^2\right\}}{p_1 \exp\left\{-\frac{z^2}{2}\right\}+\frac{p_2}{\sigma}\exp\left\{-\frac{(z\tau+\mu_1-\mu_2)^2}{2\tau^2\sigma^2}\right\}}
d z \right. \nonumber \\
				&\times \left.{}  \int_{-\infty}^\infty
\frac{\left(u^2-1\right)^2 \exp\left\{-u^2\right\}}{p_1\sigma \exp\left\{-\frac{(u\tau\sigma+\mu_2-\mu_1)^2}{2\tau^2}\right\}+p_2\exp\left\{-\frac{u^2}{2}\right\}} 
d u \right\}^\frac{1}{2} d \tau d \sigma d \mu_1 d \mu_2 d p_1 \nonumber
\end{align*}

When considering the integrals relative to the Jeffreys prior, they do not represent an issue for convergence with respect to the scale parameters, because exponential terms going to $0$ as the scale parameters tend to $0$ are present. However, when considering the part out of the previous integrals, a factor $\sigma^-2$ whose behavior is not convergent is present. Then this particular part of the posterior distribution is not integrating.

When considering the case of $k$ components, the integral will always inversily depends on $\sigma_1, \sigma_2,\cdots, \sigma_{k-1}$ and then the posterior will always be improper. 

As a note aside, it is worth noting that the usual separation between parameters proposed by Jeffreys himself in the multidimensional problems does not change the behavior of the posterior, because even if the Fisher information matrix is decomposed as

\begin{equation*}
I(\theta)=\left( \begin{matrix} I_1(\theta_1) & 0 \\ 0 & I_2(\theta_2) \end{matrix} \right)
\end{equation*}

\noindent for any possible combination of the parameters $\theta=(p,\mu_1,\mu_2,\sigma_1,\sigma_2)$ (note that $\theta_1$ and $\theta_2$ are vectors and $I(\theta_1)$ and $I(\theta_2)$ are diagonal or non-diagonal matrices), the product of the elements in the diagonal (considered in the proof) will be the same.

\end{proof}

A comparison with maximum likelihood estimation obtained via EM has shown that the Bayesian estimates obtained via MCMC and by using a Jeffreys prior seems to better identify the true values which have generated the data for a sufficient high sample size. Table \ref{tab:EMvsBayes} shows the comparison between the ML and the Bayesian estimates (for repeated simulations, the initial values for the MCMC algorithm have been randomly chosen to have a sufficiently high likelihood level). The log-likelihood value of the ML estimates is always lower that the log-likelihood value of the Bayesian estimates. The better performance of the Bayesian algorithm is only shown for practical reasons, since we have already proved the posterior distribution is improper. Figure \ref{priorlik} shows that the MCMC algorithm accepts moves with an increasing likelihood value, until this value stabilizes around $-210$. The same happens for the prior level. 

\begin{table}
\centering
    \caption{Comparison between ML estimates and Bayesian estimates obtained by using a Jeffreys prior for a 3-components Gaussian mixture model.} \label{tab:EMvsBayes}
    \begin{tabular}{ | c | c | c | c | c | c |}
    \hline
    \textbf{Parameters} & \textbf{ML} & \textbf{Bayes}  & \textbf{True} \\ \hline
     $\mu$ & (-7.245,13.308,14.999) & (-10.003,0.307,14.955) & (-10,0,15) \\ \hline
	 $\sigma$ & (0.547,5.028,0.154) & (1.243,3.642,0.607) & (1.0,5.0,0.5) \\ \hline    
	 w & (0.350,0.016,0.634) & (0.258,0.106,0.636) & (0.25,0.10,0.65) \\ \hline
    \end{tabular}
\end{table}

\begin{figure}
\centering
\includegraphics[width=6.5cm, height=7.5cm]{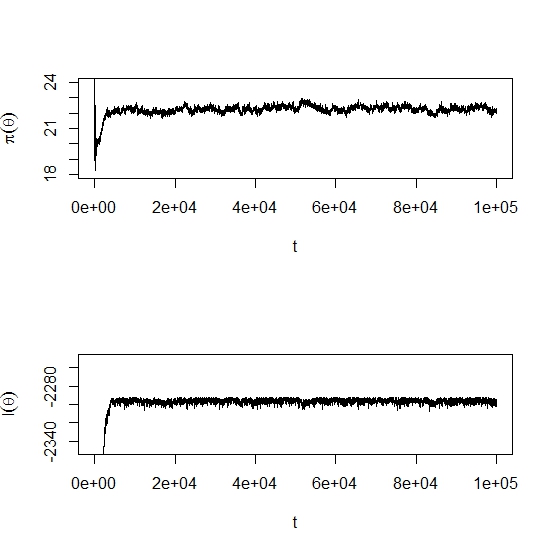}
\caption{Values of (Jeffreys) prior (above) and likelihood function (below) for the accepted moves of the MCMC algorithm which estimate the posterior distribution of the parameters of a 3-components Gaussian mixture model and a sample size equal to 1000.}
\label{priorlik}
\end{figure}

For small sample sizes, the chains tend to get stuck when very small values of standard deviations are accepted. Table \ref{tab:allunkn2} and \ref{tab:allunkn3} show the results for different sample sizes and different scenarios (in particular, the situations when the means are close or far from each other are considered) for a mixture model with two and three components respectively. The second and the third columns show the reason why the chain goes into trouble: sometimes the chains do not converge and tend towards very high values of means, sometimes the chains get stuck to very small values of standard deviations. 

\begin{table}[]
\centering
\small
\caption{k=2, ($\mathbf{p}$, $\mu$, $\sigma$) unknown: results of 50 replications of the experiment for both close and far means with a Monte Carlo 
approximation of the posterior distribution based on $10^5$ simulations and a burn-in of 
$10^4$ simulations. The table shows the average acceptance rate, the proportion of chains 
diverging towards higher values, the proportion of chains stuck at values of standard deviations close to 0 and the average ratio between the log-likelihood of the last accepted values and the true values in the 50 replications when using the Jeffreys prior.}
\label{tab:allunkn2}
\begin{tabular}{|ccccc|}
\hline
\multicolumn{1}{|l|}{{\bf k=2}}                                                    & \multicolumn{1}{l}{{\bf Jeffreys prior}}                             & \multicolumn{1}{l}{(Close Means)}                                                             & \multicolumn{1}{l}{}                                                                   & \multicolumn{1}{l|}{}                                                                                \\ \hline
\multicolumn{1}{|c|}{{\it \begin{tabular}[c]{@{}c@{}}Sample \\ Size\end{tabular}}} & {\it \begin{tabular}[c]{@{}c@{}}Ave. \\ Accept.\\ Rate\end{tabular}} & {\it \begin{tabular}[c]{@{}c@{}}Chains stuck\\ at small values\\ of $\sigma$\end{tabular}} & {\it \begin{tabular}[c]{@{}c@{}}Chains towards\\ high values of \\ $\mu$\end{tabular}} & {\it \begin{tabular}[c]{@{}c@{}}Ave. \\ lik($\theta^{fin}$)\\ /\\ lik($\theta^{true}$)\end{tabular}} \\ \hline
\multicolumn{1}{|c|}{5}                                                            & 0.1119                                                               & 0.54                                                                                          & 0.74                                                                                   & 3.5280                                                                                             \\
\multicolumn{1}{|c|}{6}                                                            & 0.1241                                                               & 0.56                                                                                          & 0.74                                                                                   & 3.6402                                                                                             \\
\multicolumn{1}{|c|}{7}                                                            & 0.0927                                                               & 0.56                                                                                          & 0.70                                                                                   & 3.2180                                                                                             \\
\multicolumn{1}{|c|}{8}                                                            & 0.0693                                                               & 0.54                                                                                          & 0.70                                                                                   & 3.1380                                                                                               \\
\multicolumn{1}{|c|}{9}                                                            & 0.1236                                                               & 0.42                                                                                          & 0.72                                                                                   & 3.3281                                                                                               \\
\multicolumn{1}{|c|}{10}                                                           & 0.1081                                                               & 0.44                                                                                          & 0.84                                                                                   & 2.8173                                                                                               \\
\multicolumn{1}{|c|}{11}                                                           & 0.1172                                                               & 0.40                                                                                          & 0.78                                                                                   & 2.1455                                                                                               \\
\multicolumn{1}{|c|}{12}                                                           & 0.1107                                                               & 0.40                                                                                          & 0.70                                                                                   & 1.8998                                                                                              \\
\multicolumn{1}{|c|}{13}                                                           & 0.1273                                                               & 0.44                                                                                          & 0.74                                                                                   & 1.8269                                                                                               \\
\multicolumn{1}{|c|}{14}                                                           & 0.1253                                                               & 0.42                                                                                          & 0.76                                                                                   & 1.2876                                                                                               \\
\multicolumn{1}{|c|}{15}                                                           & 0.1218                                                               & 0.36                                                                                          & 0.82                                                                                   & 1.2949                                                                                               \\
\multicolumn{1}{|c|}{20}                                                           & 0.1278                                                               & 0.38                                                                                          & 0.66                                                                                   & 1.2587                                                                                               \\ \hline
{\bf k=2}                                                                          & \multicolumn{1}{|l}{{\bf }}                                           & \multicolumn{1}{l}{(Far means)}                                                               & \multicolumn{1}{l}{}                                                                   & \multicolumn{1}{l|}{}                                                                                \\ \hline
\multicolumn{1}{|c|}{5}                                                            & 0.1650                                                               & 0.18                                                                                          & 0.30                                                                                   & 3.7712                                                                                              \\
\multicolumn{1}{|c|}{6}                                                            & 0.2218                                                               & 0.12                                                                                          & 0.20                                                                                   & 3.1400                                                                                               \\
\multicolumn{1}{|c|}{7}                                                            & 0.1836                                                               & 0.12                                                                                          & 0.36                                                                                   & 3.1461                                                                                               \\
\multicolumn{1}{|c|}{8}                                                            & 0.2313                                                               & 0.08                                                                                          & 0.08                                                                                   & 3.5102                                                                                               \\
\multicolumn{1}{|c|}{9}                                                            & 0.1942                                                               & 0.14                                                                                          & 0.12                                                                                   & 3.5585                                                                                               \\
\multicolumn{1}{|c|}{10}                                                           & 0.2290                                                               & 0.04                                                                                          & 0.02                                                                                   & 3.0718                                                                                               \\
\multicolumn{1}{|c|}{11}                                                           & 0.2320                                                               & 0.04                                                                                          & 0.02                                                                                   & 2.9825                                                                                               \\
\multicolumn{1}{|c|}{12}                                                           & 0.2305                                                               & 0.08                                                                                          & 0.02                                                                                   & 2.9122                                                                                               \\
\multicolumn{1}{|c|}{13}                                                           & 0.2264                                                               & 0.06                                                                                          & 0.00                                                                                   & 2.9571                                                                                               \\
\multicolumn{1}{|c|}{14}                                                           & 0.2292                                                               & 0.08                                                                                          & 0.04                                                                                   & 1.0612                                                                                               \\
\multicolumn{1}{|c|}{15}                                                           & 0.2005                                                               & 0.12                                                                                          & 0.04                                                                                   & 1.0804                                                                                               \\
\multicolumn{1}{|c|}{20}                                                           & 0.2343                                                               & 0.00                                                                                          & 0.02                                                                                   & 1.0146                                                                                               \\ \hline
\end{tabular}
\end{table}

\begin{table}[]
\centering
\caption{k=3, ($\mathbf{p}$, $\mu$, $\sigma$) unknown: as in table \ref{tab:allunkn2} for two three-component Gaussian mixture models with close and far means.}
\label{tab:allunkn3}
\begin{tabular}{|ccccc|}
\hline
\multicolumn{1}{|l|}{{\bf k=3}}                                                    & \multicolumn{1}{l}{{\bf Jeffreys prior}}                             & \multicolumn{1}{l}{(Close Means)}                                                          & \multicolumn{1}{l}{}                                                                   & \multicolumn{1}{l|}{}                                                                                \\ \hline
\multicolumn{1}{|c|}{{\it \begin{tabular}[c]{@{}c@{}}Sample \\ Size\end{tabular}}} & {\it \begin{tabular}[c]{@{}c@{}}Ave. \\ Accept.\\ Rate\end{tabular}} & {\it \begin{tabular}[c]{@{}c@{}}Chains stuck\\ at small values\\ of $\sigma$\end{tabular}} & {\it \begin{tabular}[c]{@{}c@{}}Chains towards\\ high values of \\ $\mu$\end{tabular}} & {\it \begin{tabular}[c]{@{}c@{}}Ave. \\ lik($\theta^{fin}$)\\ /\\ lik($\theta^{true}$)\end{tabular}} \\ \hline
\multicolumn{1}{|c|}{5}                                                            & 0.0302                                                               & 0.76                                                                                       & 0.44                                                                                   & 2.9095                                                                                           \\
\multicolumn{1}{|c|}{6}                                                            & 0.0368                                                               & 0.76                                                                                       & 0.48                                                                                   & 3.2507                                                                                           \\
\multicolumn{1}{|c|}{7}                                                            & 0.0290                                                               & 0.80                                                                                       & 0.30                                                                                   & 3.1318                                                                                             \\
\multicolumn{1}{|c|}{8}                                                            & 0.0578                                                               & 0.62                                                                                       & 0.54                                                                                   & 3.0043                                                                                               \\
\multicolumn{1}{|c|}{9}                                                            & 0.0488                                                               & 0.74                                                                                       & 0.52                                                                                   & 2.5798                                                                                             \\
\multicolumn{1}{|c|}{10}                                                           & 0.0426                                                               & 0.70                                                                                       & 0.44                                                                                   & 2.3023                                                                                               \\
\multicolumn{1}{|c|}{11}                                                           & 0.0572                                                               & 0.66                                                                                       & 0.38                                                                                   & 1.7497                                                                                               \\
\multicolumn{1}{|c|}{12}                                                           & 0.0464                                                               & 0.66                                                                                       & 0.48                                                                                   & 1.4032                                                                                             \\
\multicolumn{1}{|c|}{13}                                                           & 0.0706                                                               & 0.52                                                                                       & 0.44                                                                                   & 1.9303                                                                                               \\
\multicolumn{1}{|c|}{14}                                                           & 0.0556                                                               & 0.66                                                                                       & 0.36                                                                                   & 1.3588                                                                                               \\
\multicolumn{1}{|c|}{15}                                                           & 0.0610                                                               & 0.74                                                                                       & 0.44                                                                                   & 1.3588                                                                                               \\
\multicolumn{1}{|c|}{20}                                                           & 0.0654                                                               & 0.48                                                                                       & 0.46                                                                                   & 1.2161                                                                                               \\ \hline
{\bf k=3}                                                                          & \multicolumn{1}{|l}{{\bf }}                                           & \multicolumn{1}{l}{(Far means)}                                                            & \multicolumn{1}{l}{}                                                                   & \multicolumn{1}{l|}{}                                                                                \\ \hline
\multicolumn{1}{|c|}{5}                                                            & 0.0644                                                               & 0.60                                                                                       & 0.10                                                                                   & 5.9707                                                                                     \\
\multicolumn{1}{|c|}{6}                                                            & 0.0631                                                               & 0.64                                                                                       & 0.18                                                                                   & 2.0557                                                                                       \\
\multicolumn{1}{|c|}{7}                                                            & 0.0726                                                               & 0.54                                                                                       & 0.08                                                                                   & 2.9351                                                                                          \\
\multicolumn{1}{|c|}{8}                                                            & 0.1745                                                               & 0.22                                                                                       & 0.12                                                                                   & 2.9193                                                                                              \\
\multicolumn{1}{|c|}{9}                                                            & 0.1809                                                               & 0.32                                                                                       & 0.04                                                                                   & 95.793                                                                                              \\
\multicolumn{1}{|c|}{10}                                                           & 0.1724                                                               & 0.28                                                                                       & 0.14                                                                                   & 2.5938                                                                                               \\
\multicolumn{1}{|c|}{11}                                                           & 0.1948                                                               & 0.24                                                                                       & 0.14                                                                                   & 3.1566                                                                                             \\
\multicolumn{1}{|c|}{12}                                                           & 0.1718                                                               & 0.26                                                                                       & 0.08                                                                                   & 2.8595                                                                                               \\
\multicolumn{1}{|c|}{13}                                                           & 0.2110                                                               & 0.16                                                                                       & 0.06                                                                                   & 1.8595                                                                                               \\
\multicolumn{1}{|c|}{14}                                                           & 0.1880                                                               & 0.24                                                                                       & 0.10                                                                                   & 1.2165                                                                                               \\
\multicolumn{1}{|c|}{15}                                                           & 0.1895                                                               & 0.20                                                                                       & 0.12                                                                                   & 1.2133                                                                                               \\
\multicolumn{1}{|c|}{20}                                                           & 0.2468                                                               & 0.08                                                                                       & 0.02                                                                                   & 1.0146                                                                                               \\ \hline
\end{tabular}
\end{table}

Since the improperness of the posterior distribution is due to the scale parameters, we may use a reparametrization of the problem as in Equation \eqref{reparMix} and use a proper prior on the parameter $\sigma$, for example, by following \cite{robert:mengersen:1999}

\begin{equation*}
	p(\sigma)=\frac{1}{2}\mathcal{U}_{[0,1]}(\sigma)+\frac{1}{2}\frac{1}{\mathcal{U}_{[0,1]}(\sigma)}.
\end{equation*}

\noindent and the Jeffreys prior for all the other parameters $(\mathbf{p},\mu,\delta,\tau)$ conditionally on $\sigma$. 

Since the improperness of the posterior distribution is due to the scale parameters, we may use a reparametrization of the problem as in Equation \eqref{reparMix} and use a proper prior on the parameter $\sigma$, for example, by following \cite{robert:mengersen:1999}

\begin{equation*}
	p(\sigma)=\frac{1}{2}\mathcal{U}_{[0,1]}(\sigma)+\frac{1}{2}\frac{1}{\mathcal{U}_{[0,1]}(\sigma)}.
\end{equation*}

\noindent and the Jeffreys prior for all the other parameters $(\mathbf{p},\mu,\delta,\tau)$ conditionally on $\sigma$. 

Actually, using a proper prior on $\sigma$ does not avoid convergence trouble, as demonstrated by Table \ref{tab:sigmaprop2}, which shows that, even if the chains with respect to the standard deviations are not stuck around $0$ when using a proper prior for $\sigma$ in the reparametrization proposed by \cite{robert:mengersen:1999}, the chains with respect to the locations parameters demonstrate a divergent behavior.

\begin{table}[]
\centering
\small
\caption{k=2, ($\mathbf{p}$, $\mu$, $\delta$, $\tau$,$\sigma$) unknown, proper prior on $\sigma$: results of 50 replications of the experiment by using a proper prior on $\sigma$ and the Jeffreys prior for the other parameters conditionally on it for both close and far means with a Monte Carlo 
approximation of the posterior distribution based on $10^5$ simulations and a burn-in of 
$10^4$ simulations. The table shows the average acceptance rate, the proportion of chains 
diverging towards higher values, the proportion of chains stuck at values of standard deviations close to 0 and the average ratio between the log-likelihood of the last accepted values and the true values in the 50 replications when using the Jeffreys prior. }
\label{tab:sigmaprop2}
\begin{tabular}{|c|cccc|}
\hline
\multicolumn{1}{|l|}{{\bf K=2}}                              & \multicolumn{1}{l}{{\bf Jeffreys prior}}                             & \multicolumn{1}{l}{(Close Means)}                                                          & \multicolumn{1}{l}{}                                                                   & \multicolumn{1}{l|}{}                                                                                \\ \hline
{\it \begin{tabular}[c]{@{}c@{}}Sample \\ Size\end{tabular}} & {\it \begin{tabular}[c]{@{}c@{}}Ave. \\ Accept.\\ Rate\end{tabular}} & {\it \begin{tabular}[c]{@{}c@{}}Chains stuck\\ at small values\\ of $\sigma$\end{tabular}} & {\it \begin{tabular}[c]{@{}c@{}}Chains towards\\ high values of \\ $\mu$\end{tabular}} & {\it \begin{tabular}[c]{@{}c@{}}Ave. \\ lik($\theta^{fin}$)\\ /\\ lik($\theta^{true}$)\end{tabular}} \\ \hline
5                                                            & 0.2094                                                               & 0.02                                                                                       & 0.92                                                                                   & 1.4440                                                                                               \\
6                                                            & 0.2152                                                               & 0.00                                                                                       & 0.98                                                                                   & 1.3486                                                                                               \\
7                                                            & 0.2253                                                               & 0.00                                                                                       & 0.92                                                                                   & 1.3290                                                                                               \\
8                                                            & 0.2021                                                               & 0.00                                                                                       & 0.94                                                                                   & 1.2258                                                                                               \\
9                                                            & 0.1828                                                               & 0.00                                                                                       & 0.84                                                                                   & 1.2666                                                                                               \\
10                                                           & 0.2087                                                               & 0.00                                                                                       & 0.88                                                                                   & 1.1770                                                                                               \\
11                                                           & 0.1854                                                               & 0.00                                                                                       & 0.94                                                                                   & 1.2088                                                                                               \\
12                                                           & 0.1829                                                               & 0.00                                                                                       & 0.86                                                                                   & 1.2153                                                                                               \\
13                                                           & 0.1658                                                               & 0.00                                                                                       & 0.92                                                                                   & 1.1682                                                                                               \\
14                                                           & 0.2017                                                               & 0.00                                                                                       & 0.86                                                                                   & 1.2043                                                                                               \\
15                                                           & 0.1991                                                               & 0.00                                                                                       & 0.88                                                                                   & 1.2002                                                                                               \\
20                                                           & 0.1851                                                               & 0.00                                                                                       & 0.76                                                                                   & 1.1688                                                                                               \\ \hline
{\bf K=2}                                                    & \multicolumn{1}{l}{{\bf }}                                           & \multicolumn{1}{l}{(Far means)}                                                            & \multicolumn{1}{l}{}                                                                   & \multicolumn{1}{l|}{}                                                                                \\ \hline
5                                                            & 0.2071                                                               & 0.00                                                                                       & 0.70                                                                                   & 1.5741                                                                                               \\
6                                                            & 0.2021                                                               & 0.00                                                                                       & 0.68                                                                                   & 1.4384                                                                                               \\
7                                                            & 0.1947                                                               & 0.00                                                                                       & 0.60                                                                                   & 1.3597                                                                                               \\
8                                                            & 0.2054                                                               & 0.00                                                                                       & 0.44                                                                                   & 1.2869                                                                                               \\
9                                                            & 0.2093                                                               & 0.00                                                                                       & 0.46                                                                                   & 1.3064                                                                                               \\
10                                                           & 0.2271                                                               & 0.00                                                                                       & 0.20                                                                                   & 1.1618                                                                                               \\
11                                                           & 0.2030                                                               & 0.00                                                                                       & 0.32                                                                                   & 1.1996                                                                                               \\
12                                                           & 0.2178                                                               & 0.00                                                                                       & 0.24                                                                                   & 1.1494                                                                                               \\
13                                                           & 0.2812                                                               & 0.00                                                                                       & 0.18                                                                                   & 1.1215                                                                                               \\
14                                                           & 0.1880                                                               & 0.00                                                                                       & 0.08                                                                                   & 1.0717                                                                                               \\
15                                                           & 0.2511                                                               & 0.00                                                                                       & 0.06                                                                                   & 1.0594                                                                                               \\
20                                                           & 0.2359                                                               & 0.00                                                                                       & 0.00                                                                                   & 1.0166                                                                                               \\ \hline
\end{tabular}
\end{table}

\vspace{0.3cm}
\section{A noninformative alternative to Jeffreys prior}
\label{sec:alternative}

The information brought by the Jeffreys prior does not seem to be enough to conduct inference in the case of mixture
models. The computation of the determinant creates a dependence between the elements of the Fisher information matrix in
the definition of the prior distribution which makes it difficult to find slight modifications of this prior that would
lead to a proper posterior distribution. For example, using a proper prior for part of the scale parameters and the
Jeffreys prior conditionally on them does not avoid impropriety, as we have demonstrated in Section \ref{sub:post}.

The literature covers attempts to define priors which add a small amount of information that is sufficient to conduct
the statistical analysis without overwhelming the information contained in the data. Some of these are related to the
computational issues in estimating the parameters of mixture models, as in the approach of
\cite{casella:mengersen:robert:titterington:2002}, who find a way to use perfect slice sampler by focusing on components
in the exponential family and conjugate priors. A characteristic example is given by \cite{richardson:green:1997}, who
propose weakly informative priors, which are data-dependent (or
empirical Bayes) and are represented by flat normal priors over an interval corresponding to the range of the data.
Nevertheless, since mixture models belong to the class of ill-posed problems, the influence of a proper prior over the
resulting inference is difficult to assess.

Another solution found in \cite{mengersen:robert:1996} proceeds through the reparametrization \eqref{reparMix} and
introduces a reference component that allows for improper priors. This approach then envisions the other parameters 
as departures from the reference and ties them together by considering each parameter $\theta_i$ as a perturbation of
the parameter of the previous component $\theta_{i-1}$. This perspective is justified by the fact that the $(i-1)$-th
component is not informative enough to absorb all the variability in the data. For instance, a three-component 
mixture model gets rewritten as
\begin{align*}
p\mathcal{N}(\mu,\tau^2)&+(1-p)q\mathcal{N}(\mu+\tau\theta,\tau^2\sigma_1^2) \\
						&\quad {} + (1-p)(1-q)\mathcal{N}(\mu+\tau\theta+\tau\sigma\epsilon,\tau^2\sigma_1^2\sigma_2^2)
\end{align*} 
\noindent where one can impose the constraint $1 \geq \sigma_1 \geq \sigma_2$ for identifiability reasons. Under this
representation, it is possible to use an improper prior on the global location-scale parameter $(\mu,\tau)$, while
proper priors must be applied to the remaining parameters. This reparametrization has been used also for exponential components by \cite{gruet:philippe:robert:1999} and Poisson
components by \cite{robert:titterington:1998}. Moreover, \cite{roeder:wasserman:1997} propose a Markov prior which follows the same resoning of dependence between the parameters for Gaussian components, where each parameter is again a perturbation of the parameter of the previous component $\theta_{i-1}$.

This representation suggests to define a global location-scale parameter in a more implicit way, via a hierarchical
model that considers more levels in the analysis and choose noninformative priors at the last level in the hierarchy.
 
More precisely, consider the Gaussian mixture model \eqref{eq:theMix}

\begin{equation}
\label{eq:hierarc1}
g(x|\boldsymbol{\theta})=\sum_{i=1}^K p_i \mathfrak{n}(x|\mu_i,\sigma_i).
\end{equation}

The parameters of each component may be considered as related in some way; for example, the observations have a
reasonable range, which makes it highly improbable to face very different means in the above Gaussian mixture model. A
similar argument may be used for the standard deviations. 

Therefore, at the second level of the hierarchical model, we may write

\begin{align}
\label{eq:hierarc2}
\mu_i & \stackrel{iid}{\sim} \mathcal{N}(\mu_0, \zeta_0) \nonumber \\
\sigma_i & \stackrel{iid}{\sim} \frac{1}{2} \mathcal{U}(0,\zeta_0) + \frac{1}{2}\frac{1}{\mathcal{U}(0,\zeta_0)} \nonumber \\
\mathbf{p} & \sim Dir\left(\frac{1}{2},\cdots,\frac{1}{2}\right) 
\end{align}

\noindent which indicates that the location parameters vary between components, but are likely to be close, and
that the scale parameters may be lower or bigger than $\zeta_0$, but not exactly equal to $\zeta_0$. The weights are
given a Dirichlet prior (or in the case of just two components, a Beta prior) independently from the components'
parameters.

At the third level of the hierarchical model, the prior may be noninformative:

\begin{align}
\label{eq:hierarc3}
\pi(\mu_0,\zeta_0) \propto \frac{1}{\zeta_0}
\end{align}

As in \citet{mengersen:robert:1996} the parameters in the mixture model are considered tied together; on the other hand, this feature is not obtained via a representation of the mixture model itself, but via a hierarchy in the definition of the model and the parameters.   

\begin{theorem}
The posterior distribution derived from the hierarchical representation of the Gaussian mixture model 
associated with \eqref{eq:hierarc1}, \eqref{eq:hierarc2} and \eqref{eq:hierarc3}
is proper. 
\end{theorem}

\begin{proof}
Consider the composition of the three levels of the hierarchical model described in equations \eqref{eq:hierarc1}, \eqref{eq:hierarc2} and \eqref{eq:hierarc3}:

\begin{align}
\label{eq:hierarch_post}
\pi(\boldsymbol{\mu},\boldsymbol{\sigma},\mu_0,\zeta_0;\mathbf{x}) & \propto L(\mu_1,\mu_2,\sigma_1,\sigma_2;\mathbf{x})  p^{-1/2} (1-p)^{-1/2}
			\nonumber \\
			& {} \times \frac{1}{\zeta_0} \frac{1}{2\pi\zeta_0^2} \exp\left\{- \frac{(\mu_1-\mu_0)^2 (\mu_2-\mu_0)^2}{2\zeta_0^2}\right\} \nonumber \\ 
			& {} \times \left[ \frac{1}{2}\frac{1}{\zeta_0} \mathbb{I}_{[\sigma_1\in(0,\zeta_0)]}(\sigma_1) + \frac{1}{2}\frac{\zeta_0}{\sigma_1^2} \mathbb{I}_{[\sigma_1\in(\zeta_0,+\infty)]}(\sigma_1) \right] \nonumber \\
			& {} \times \left[ \frac{1}{2}\frac{1}{\zeta_0} \mathbb{I}_{[\sigma_2\in(0,\zeta_0)]}(\sigma_2) + \frac{1}{2}\frac{\zeta_0}{\sigma_2^2} \mathbb{I}_{[\sigma_2\in(\zeta_0,+\infty)]}(\sigma_2) \right]
\end{align}

\noindent where $L(\cdot;\mathbf{x})$ is given by Equation \eqref{eq:mixlik}. 

Once again, we can initialize the proof by considering only the first term in the sum composing the likelihood function
for the mixture model. Then the product in \eqref{eq:hierarch_post} may be split into four terms corresponding to the
different terms in the scale parameters' prior. For instance, the first term is

\begin{align*}
\int_0^\infty & \int_{-\infty}^\infty \int_\mathbb{R}\int_\mathbb{R} \int_\mathbb{R^+} \int_\mathbb{R^+} \int_0^1 
				\frac{1}{\sigma_1^n} p_1^n \exp \left\{- \frac{\sum_{i=1}^n (x_i-\mu_1)^2}{2\sigma_1^2} \right\} \nonumber \\
				& {} \times \frac{1}{\zeta_0^3} \exp\left\{-\frac{(\mu_1-\mu_0)^2 (\mu_2-\mu_0)^2}{2\zeta_0^2} \right\} \nonumber \\
				& {} \times \frac{1}{4}\frac{1}{\zeta_0} \frac{1}{\zeta_0} \mathbb{I}_{[\sigma_1 \in (0,\zeta_0)]}(\sigma_1) \mathbb{I}_{[\sigma_2 \in (0,\zeta_0)]}(\sigma_2)
d p d\sigma_1 d\sigma_2 d\mu_1 d\mu_2 d \mu_0 d\zeta_0 
\end{align*}

\noindent and the second one

\begin{align*}
\int_0^\infty & \int_{-\infty}^\infty \int_\mathbb{R}\int_\mathbb{R} \int_\mathbb{R^+} \int_\mathbb{R^+} \int_0^1 
				\frac{1}{\sigma_1^n} p_1^n \exp \left\{- \frac{\sum_{i=1}^n (x_i-\mu_1)^2}{2\sigma_1^2} \right\} \nonumber \\
				& {} \times \frac{1}{\zeta_0^3} \exp\left\{-\frac{(\mu_1-\mu_0)^2 (\mu_2-\mu_0)^2}{2\zeta_0^2} \right\} \nonumber \\
				& {} \times \frac{1}{4}\frac{1}{\zeta_0} \frac{\zeta_0}{\sigma_2^2} \mathbb{I}_{[\sigma_1 \in (0,\zeta_0)]}(\sigma_1) \mathbb{I}_{[\sigma_2 \in (\zeta_0,\infty)]}(\sigma_2)
d p d\sigma_1 d\sigma_2 d\mu_1 d\mu_2 d \mu_0 d\zeta_0 .
\end{align*}

The integrals with respect to $\mu_1$, $\mu_2$ and $\mu_0$ converge, since the data are carrying
information about $\mu_0$ through $\mu_1$. The integral with respect to $\sigma_1$ converges as well, because, as
$\sigma_1 \rightarrow 0$, the exponential function goes to $0$ faster than $\frac{1}{\sigma_1^n}$ goes to $\infty$
(integrals where $\sigma_1>\zeta_0$ are not considered here because this reasoning may easily extend to those
cases).  The integrals with respect to $\sigma_2$ 
converge, because they provide a factor proportional to $\zeta_0$ and $1/\zeta_0$ respectively
which simplifies with the normalizing constant of the reference distribution (the uniform in the first case and the
Pareto in second one). Finally, the term $1/\zeta_0^4$ resulting from the previous operations has its counterpart in the
integrals relative to the location priors. Therefore, the integral with respect to $\zeta_0$ converges. 

The part of the posterior distribution relative to the weights is not an issue, since the weights belong to the
corresponding simplex.
\end{proof}

Table \ref{tab:hierMM} shows the results given by simulation from the posterior distribution of the hierarchical mixture
model and confirms that the chains always converge. 

\begin{table}[]
\centering
\footnotesize
\caption{Hierarchical Mixture model: results of 50 replications of the experiment for a two and a three-component Gaussian mixture model with a Monte Carlo approximation of the posterior distribution based on $10^5$ simulations and a burn-in of $10^4$ simulations. The table shows the average acceptance rate, the proportion of chains diverging towards higher values, the proportion of chains stuck at values of standard deviations close to 0, the mean and the median log-likelihood of the last accepted values and the mean and the median maximum log-likelihood of the accepted values.}
\label{tab:hierMM}
\begin{tabular}{|c|ccccccc|}
\hline
\multicolumn{1}{|l|}{{\bf k=2}}                              & \multicolumn{1}{l}{}                                                 & \multicolumn{1}{l}{}                                                                       & \multicolumn{1}{l}{}                                                                   & \multicolumn{1}{l}{}                                                                         & \multicolumn{1}{l}{}                                                                            & \multicolumn{1}{l}{}                                                                        & \multicolumn{1}{l|}{}                                                                         \\ \hline
{\it \begin{tabular}[c]{@{}c@{}}Sample \\ Size\end{tabular}} & {\it \begin{tabular}[c]{@{}c@{}}Ave. \\ Accept.\\ Rate\end{tabular}} & {\it \begin{tabular}[c]{@{}c@{}}Chains stuck\\ at small values\\ of $\sigma$\end{tabular}} & {\it \begin{tabular}[c]{@{}c@{}}Chains towards\\ high values of \\ $\mu$\end{tabular}} & {\it \begin{tabular}[c]{@{}c@{}}Mean\\ l($\theta^{fin}$)/\\ l($\theta^{true}$)\end{tabular}} & {\it \begin{tabular}[c]{@{}c@{}}Median \\ l($\theta^{fin}$)/\\ l($\theta^{true}$)\end{tabular}} & {\it \begin{tabular}[c]{@{}c@{}}Mean\\ max(l($\theta$))/\\ l($\theta^{true}$)\end{tabular}} & {\it \begin{tabular}[c]{@{}c@{}}Median\\ max(l($\theta$))/\\ l($\theta^{true}$)\end{tabular}} \\ \hline
3                                                            & 0.1947                                                               & 0.00                                                                                       & 0.00                                                                                   & 1.1034                                                                                       & 0.9825                                                                                          & 0.0838                                                                                      & 0.5778                                                                                        \\
4                                                            & 0.2295                                                               & 0.00                                                                                       & 0.00                                                                                   & 1.0318                                                                                       & 1.0300                                                                                          & 0.4678                                                                                      & 0.5685                                                                                        \\
5                                                            & 0.2230                                                               & 0.00                                                                                       & 0.00                                                                                   & 0.9572                                                                                       & 0.9924                                                                                          & 0.8464                                                                                      & 0.7456                                                                                        \\
6                                                            & 0.2275                                                               & 0.00                                                                                       & 0.00                                                                                   & 0.9870                                                                                       & 0.9641                                                                                          & 0.6614                                                                                      & 0.6696                                                                                        \\
7                                                            & 0.2112                                                               & 0.00                                                                                       & 0.00                                                                                   & 1.0658                                                                                       & 1.0043                                                                                          & 0.8406                                                                                      & 0.7848                                                                                        \\
8                                                            & 0.2833                                                               & 0.00                                                                                       & 0.00                                                                                   & 1.0077                                                                                       & 1.0284                                                                                          & 0.8268                                                                                      & 0.8495                                                                                        \\
9                                                            & 0.2696                                                               & 0.00                                                                                       & 0.00                                                                                   & 1.0741                                                                                       & 1.0179                                                                                          & 0.8854                                                                                      & 0.8613                                                                                        \\
10                                                           & 0.2266                                                               & 0.00                                                                                       & 0.00                                                                                   & 1.1446                                                                                       & 0.9968                                                                                          & 0.9589                                                                                      & 0.8508                                                                                        \\
15                                                           & 0.1982                                                               & 0.00                                                                                       & 0.00                                                                                   & 1.0201                                                                                       & 0.9959                                                                                          & 0.9409                                                                                      & 0.9280                                                                                        \\
20                                                           & 0.2258                                                               & 0.00                                                                                       & 0.00                                                                                   & 1.2023                                                                                       & 1.0145                                                                                          & 0.9172                                                                                      & 0.9400                                                                                        \\
30                                                           & 0.2073                                                               & 0.00                                                                                       & 0.00                                                                                   & 0.9888                                                                                       & 1.0022                                                                                          & 1.0424                                                                                      & 0.9656                                                                                        \\
50                                                           & 0.2724                                                               & 0.00                                                                                       & 0.00                                                                                   & 1.0493                                                                                       & 1.0043                                                                                          & 1.0281                                                                                      & 0.9859                                                                                        \\
100                                                          & 0.2739                                                               & 0.00                                                                                       & 0.00                                                                                   & 1.0932                                                                                       & 1.0025                                                                                          & 1.0805                                                                                      & 0.9932                                                                                        \\
200                                                          & 0.3031                                                               & 0.00                                                                                       & 0.00                                                                                   & 1.1610                                                                                       & 1.0036                                                                                          & 1.1519                                                                                      & 0.9964                                                                                        \\
500                                                          & 0.2753                                                               & 0.00                                                                                       & 0.00                                                                                   & 1.1729                                                                                       & 1.0023                                                                                          & 1.1694                                                                                      & 0.9989                                                                                        \\
1000                                                         & 0.2317                                                               & 0.00                                                                                       & 0.00                                                                                   & 1.1800                                                                                       & 1.0021                                                                                          & 1.1772                                                                                      & 0.9994                                                                                        \\ \hline
{\bf k=3}                                                    & \multicolumn{1}{l}{}                                                 & \multicolumn{1}{l}{}                                                                       & \multicolumn{1}{l}{}                                                                   & \multicolumn{1}{l}{}                                                                         & \multicolumn{1}{l}{}                                                                            & \multicolumn{1}{l}{}                                                                        & \multicolumn{1}{l|}{}                                                                         \\ \hline
3                                                            & 0.2840                                                               & 0.00                                                                                       & 0.00                                                                                   & 1.1316                                                                                       & 1.0503                                                                                          & 0.3432                                                                                      & 0.2950                                                                                        \\
4                                                            & 0.2217                                                               & 0.00                                                                                       & 0.00                                                                                   & 1.0326                                                                                       & 0.9452                                                                                          & 0.6699                                                                                      & 0.6624                                                                                        \\
5                                                            & 0.2144                                                               & 0.00                                                                                       & 0.00                                                                                   & 1.0610                                                                                       & 1.0421                                                                                          & 0.6858                                                                                      & 0.6838                                                                                        \\
6                                                            & 0.2258                                                               & 0.00                                                                                       & 0.00                                                                                   & 1.0908                                                                                       & 0.9683                                                                                          & 0.6472                                                                                      & 0.6355                                                                                        \\
7                                                            & 0.1843                                                               & 0.00                                                                                       & 0.00                                                                                   & 1.0436                                                                                       & 0.9915                                                                                          & 0.7878                                                                                      & 0.8008                                                                                        \\
8                                                            & 0.2760                                                               & 0.00                                                                                       & 0.00                                                                                   & 1.0276                                                                                       & 1.0077                                                                                          & 0.7996                                                                                      & 0.7958                                                                                        \\
9                                                            & 0.2028                                                               & 0.00                                                                                       & 0.00                                                                                   & 1.0025                                                                                       & 1.0145                                                                                          & 0.7830                                                                                      & 0.8016                                                                                        \\
10                                                           & 0.2116                                                               & 0.00                                                                                       & 0.00                                                                                   & 1.0426                                                                                       & 1.0015                                                                                          & 0.8752                                                                                      & 0.8591                                                                                        \\
15                                                           & 0.2023                                                               & 0.00                                                                                       & 0.00                                                                                   & 1.0247                                                                                       & 1.0063                                                                                          & 0.8810                                                                                      & 0.8871                                                                                        \\

20                                                           & 0.2211                                                               & 0.00                                                                                       & 0.00                                                                                   & 1.0281                                                                                       & 1.0104                                                                                          & 0.9290                                                                                      & 0.9268                                                                                        \\
30                                                           & 0.2242                                                               & 0.00                                                                                       & 0.00                                                                                   & 1.1978                                                                                       & 1.0123                                                                                          & 1.0841                                                                                      & 0.9508                                                                                        \\
50                                                           & 0.2513                                                               & 0.00                                                                                       & 0.00                                                                                   & 1.0543                                                                                       & 1.0142                                                                                          & 1.0148                                                                                      & 0.9775                                                                                        \\
100                                                          & 0.2768                                                               & 0.00                                                                                       & 0.00                                                                                   & 1.0563                                                                                       & 1.0206                                                                                          & 1.0324                                                                                      & 0.9955                                                                                        \\
200                                                          & 0.2910                                                               & 0.00                                                                                       & 0.00                                                                                   & 1.0325                                                                                       & 1.0118                                                                                          & 1.0200                                                                                      & 0.9993                                                                                        \\
500                                                          & 0.2329                                                               & 0.00                                                                                       & 0.00                                                                                   & 1.0943                                                                                       & 1.0079                                                                                          & 1.0882                                                                                      & 1.0002                                                                                        \\
1000                                                         & 0.2189                                                               & 0.00                                                                                       & 0.00                                                                                   & 1.1068                                                                                       & 1.0105                                                                                          & 1.1212                                                                                      & 1.0110                                                                                        \\ \hline
\end{tabular}
\end{table}

Figures \ref{fig:hierc_densmean2_3_8}--\ref{fig:hierc_densmean3_30_1000} show the results how a simulations study to approximate the posterior distribution of the means of a two or three-component mixture model, compared to the true values (red vertical lines) and for different sample sizes, from $n=3$ to $n=1000$.

\begin{figure}
\centering
\includegraphics[scale=0.5]{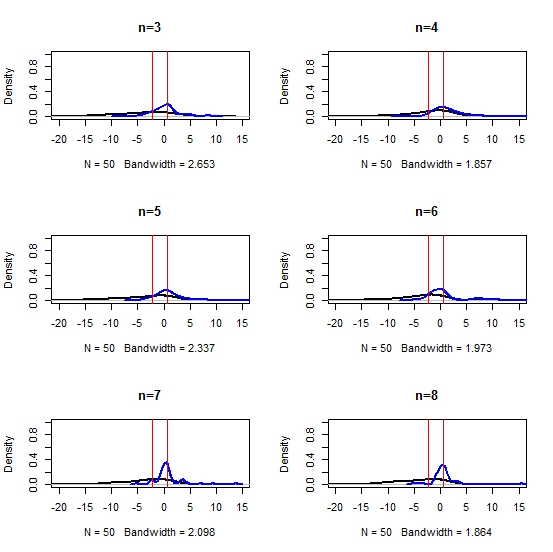}
\caption{Distribution of the posterior means for the hierarchical mixture model with two components, global mean $\mu_0=0$ and global variance $\zeta_0=5$, based on $50$ replications of the experiment with different sample sizes, black and blue lines for the marginal posterior distribution of $\mu_1$ and $\mu_2$ respectively.}
\label{fig:hierc_densmean2_3_8}
\end{figure}

\begin{figure}
\centering
\includegraphics[scale=0.5]{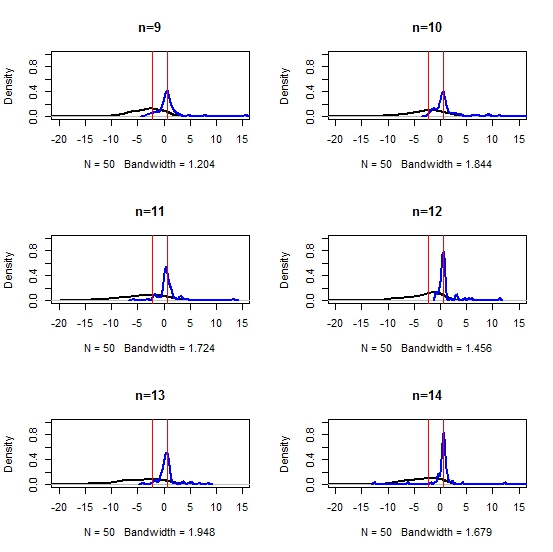}
\caption{Same caption as in Figure \ref{fig:hierc_densmean2_3_8}.}
\label{fig:hierc_densmean2_9_14}
\end{figure}

\begin{figure}
\centering
\includegraphics[scale=0.5]{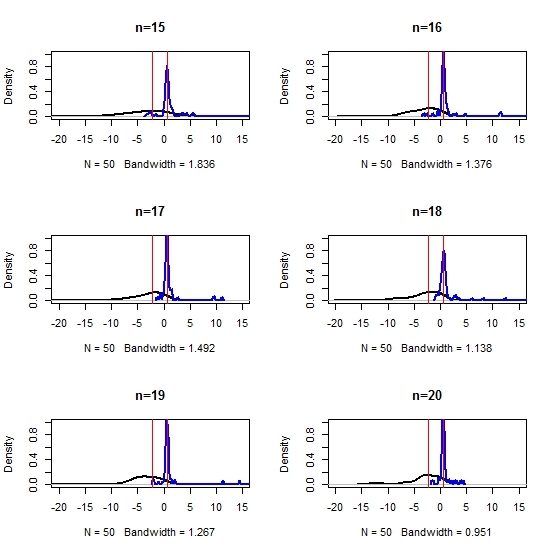}
\caption{Same caption as in Figure \ref{fig:hierc_densmean2_3_8}.}
\label{fig:hierc_densmean2_15_20}
\end{figure}

\begin{figure}
\centering
\includegraphics[scale=0.5]{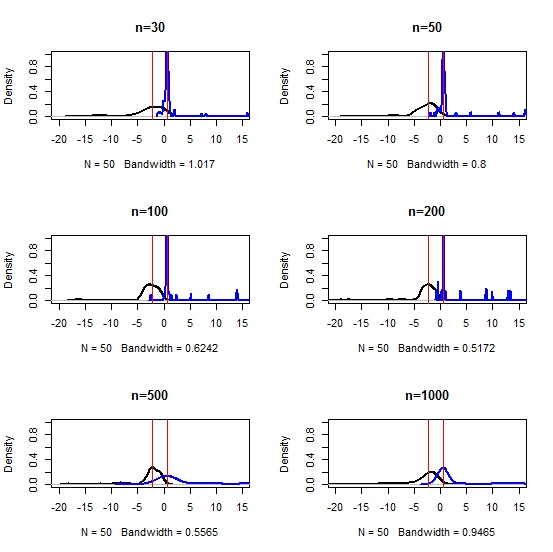}
\caption{Same caption as in Figure \ref{fig:hierc_densmean2_3_8}.}
\label{fig:hierc_densmean2_30_1000}
\end{figure}

\begin{figure}
\centering
\includegraphics[scale=0.5]{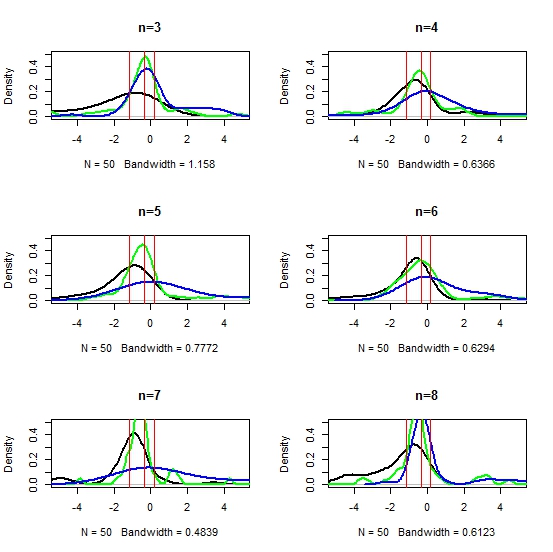}
\caption{Distribution of the posterior means for the hierarchical mixture model with three components, global mean $\mu_0=0$ and global variance $\zeta_0=5$, based on $50$ replications of the experiment with different sample sizes (the red lines stands for the true values, black, green and blue lines for the marginal posterior distributions of $\mu_1$, $\mu_2$ and $\mu_3$ respectively).}
\label{fig:hierc_densmean3_3_8}
\end{figure}

\begin{figure}
\centering
\includegraphics[scale=0.5]{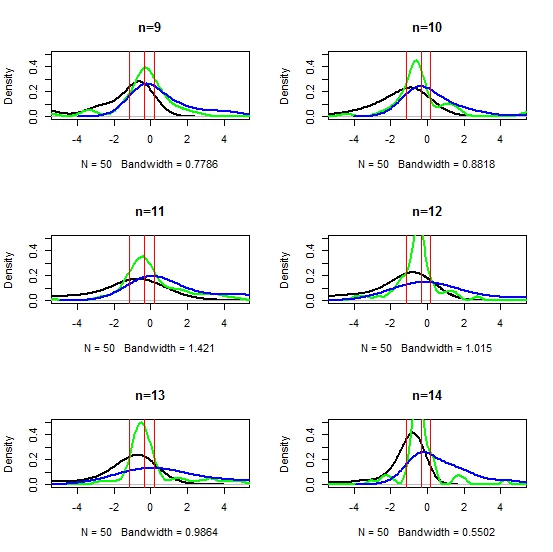}
\caption{Same caption as in Figure \ref{fig:hierc_densmean3_3_8}.}
\label{fig:hierc_densmean3_9_14}
\end{figure}

\begin{figure}
\centering
\includegraphics[scale=0.5]{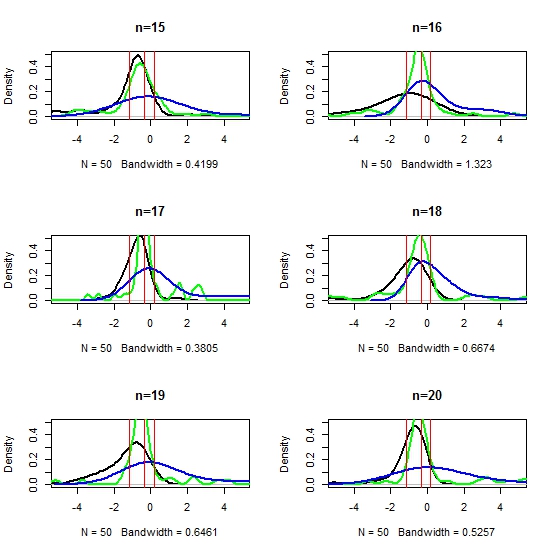}
\caption{Same caption as in Figure \ref{fig:hierc_densmean3_3_8}.}
\label{fig:hierc_densmean3_15_20}
\end{figure}

\begin{figure}
\centering
\includegraphics[scale=0.5]{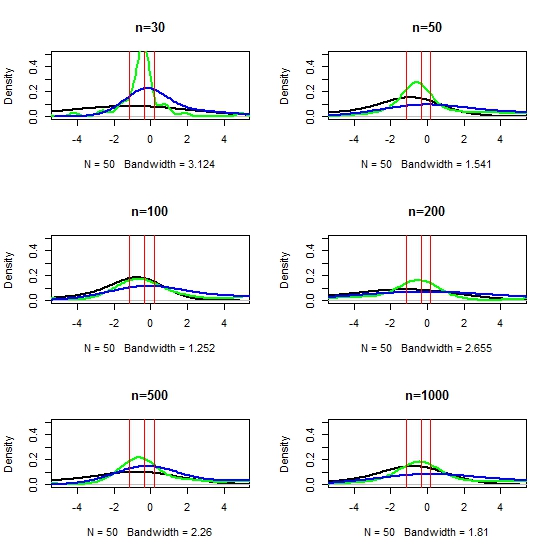}
\caption{Same caption as in Figure \ref{fig:hierc_densmean3_3_8}.}
\label{fig:hierc_densmean3_30_1000}
\end{figure}

\vspace{0.3cm}
\section{Implementation features}
\label{sec:implant}

The computing expense due to derive the Jeffreys prior for a set of parameter values is in $\mathrm{O}(d^2)$ if
$d$ is the total number of (independent) parameters.

Each element of the Fisher information matrix is an integral of the form

\begin{equation*}
-\int_{\mathcal{X}} \frac{\partial^2 \log \left[\sum_{h=1}^k p_h\,f(x|\theta_h)\right]}{\partial \theta_i \partial \theta_j}\left[\sum_{h=1}^k p_h\,f(x|\theta_h)\right]^{-1} d x
\end{equation*}

\noindent which has to be approximated. We have applied both numerical
integration and Monte Carlo integration and simulations show that, in general,
numerical integration obtained via Gauss-Kronrod quadrature (see
\cite{piessens:1983} for details), has more stable results. Neverthless, when
one or more proposed values for the standard deviations or the weights is too
small, the approximations tend to be very dependent on the bounds used for
numerical integration (usually chosen to omit a negligible part of the
density) or the numerical approximation may not be even applicable. 
In this case, Monte Carlo integration seems to have more stable, where the stability of the results depends on the Monte Carlo sample size. 

Figure \ref{fig:MCvsNUM_incrN} shows the value of the Jeffreys prior obtained via Monte Carlo integration of the
elements of the Fisher information matrix for an increasing number of Monte Carlo simulations both in the case where the
Jeffreys prior is concentrated (where the standard deviations are small) and  where it assumes low values. The value
obtained via Monte Carlo integration is then compared with the value obtained via numerical integration. The sample size
relative to the point where the graph stabilizes may be chosen to perform the approximation.

A similar analysis is shown in Figures \ref{fig:MCvsNUM_bpl1} and \ref{fig:MCvsNUM_bpl2} which provide the boxplots of $100$ replications of the Monte Carlo approximations for different numbers of simulations (on the \textit{x}-axis); one can choose to use the number of simulations which lead to a reasonable or acceptable variability of the results. 

\begin{figure}
\centering
\includegraphics[scale=0.4]{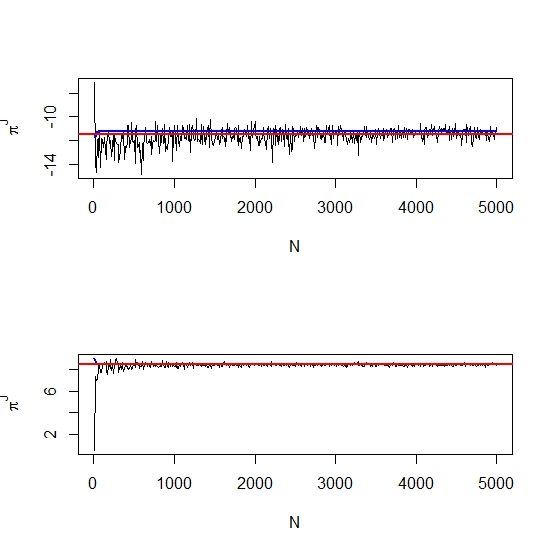}
\caption{Jeffreys prior obtained via Monte Carlo integration (and numerical integration, in \textit{red}) for the model $0.25\mathcal{N}(-10,1)+0.10\mathcal{N}(0,5)+0.65\mathcal{N}(15,7)$ (above) and for the model $\frac{1}{3}\mathcal{N}(-1,0.2)+\frac{1}{3}\mathcal{N}(0,0.2)+\frac{1}{3}\mathcal{N}(1,0.2)$ (below).}
\label{fig:MCvsNUM_incrN}
\end{figure}

\begin{figure}
\centering
\includegraphics[scale=0.4]{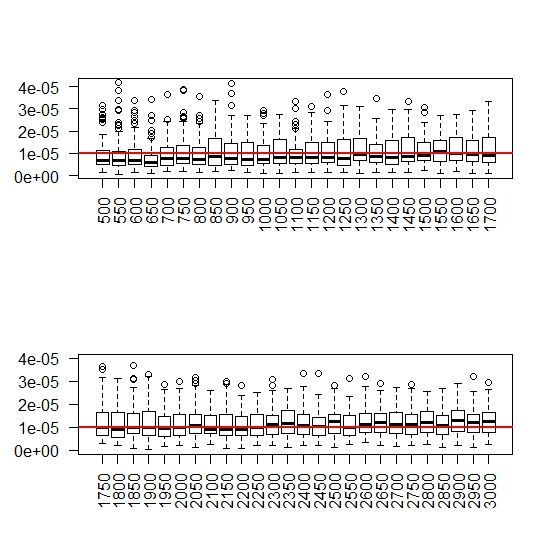}
\caption{Boxplots of 100 replications of the procedure which approximates the Fisher information matrix via Monte Carlo integration to obtain the Jeffreys prior for the model $0.25\mathcal{N}(-10,1)+0.10\mathcal{N}(0,5)+0.65\mathcal{N}(15,7)$ for sample sizes from $500$ to $3000$. The value obtained via numerical integration is represented by the red line.}
\label{fig:MCvsNUM_bpl1}
\end{figure}

\begin{figure}
\centering
\includegraphics[scale=0.4]{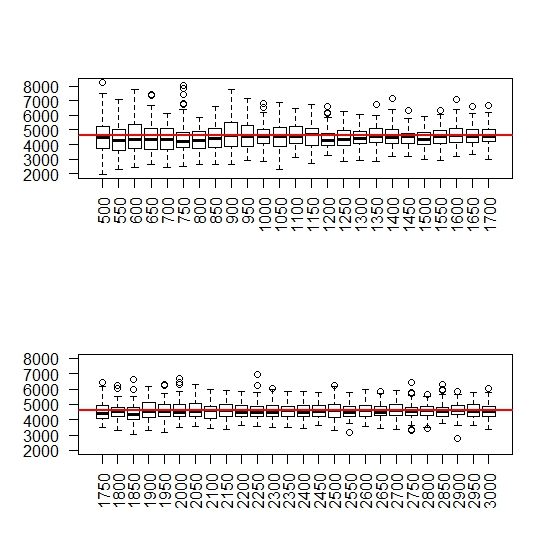}
\caption{Same caption as in Figure \ref{fig:MCvsNUM_bpl1} for the model $\frac{1}{3}\mathcal{N}(-1,0.2)+\frac{1}{3}\mathcal{N}(0,0.2)+\frac{1}{3}\mathcal{N}(1,0.2)$.}
\label{fig:MCvsNUM_bpl2}
\end{figure}

Since the approximation problem is one-dimensional, another numerical solution could be based on the sums of Riemann;
Figure \ref{fig:MCvsNUMSRvsINTR} shows the comparison between the results of the Gauss-Kronrod quadrature procedure and
a procedure based on sums of Riemann for an increasing number of points considered in a region which contain the
$99.999\%$ of the data density. Moreover, Figure \ref{fig:MCvsRIEMbxp} shows the comparison between the approximation to
the Jeffreys prior obtained via Monte Carlo integration and via the sums of Riemann: it is clear that the sums of
Riemann lead to more stable results in comparison with Monte Carlo integration. On the other hand, they can be applied
in more situations than the Gauss-Kromrod quadrature, in particular, in cases where the standard deviations are very
small (of order $10^{-2}$). Nevertheless, when the standard deviations are smaller than this, one has to pay attention
on the features of the function to integrate. In fact, the mixture density tends to concentrate around the modes, with
regions of density close to 0 between them. The elements of the Fisher informtation matrix are, in general, ratios
between the components' densities and the mixture density, then in those regions an indeterminate form of type
$\frac{0}{0}$ is obtained; Figure \ref{fig:FishInfoelem} represents the behavior of one of these elements when $\sigma_i
\rightarrow 0$ for $i=1,\cdots,k$. 

\begin{figure}
\centering
\includegraphics[scale=0.4]{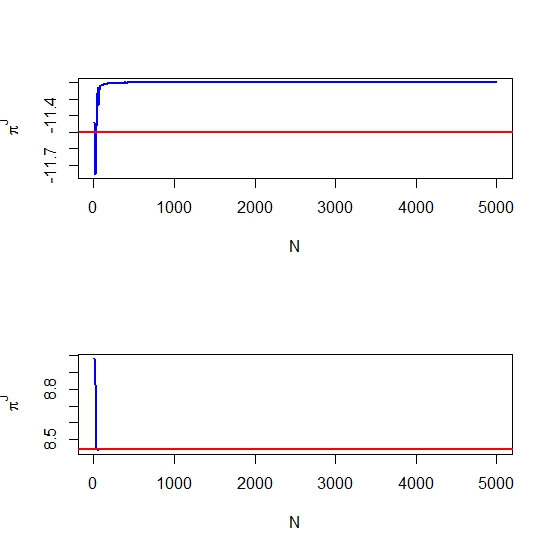}
\caption{Comparison between the Jeffreys prior density obtained via integration in the Fisher information matrix via Gauss-Kronrod quadrature and sums of Riemann for the model $0.25\mathcal{N}(-10,1)+0.10\mathcal{N}(0,5)+0.65\mathcal{N}(15,7)$ (above) and $\frac{1}{3}\mathcal{N}(-1,0.2)+\frac{1}{3}\mathcal{N}(0,0.2)+\frac{1}{3}\mathcal{N}(1,0.2)$ (below).}
\label{fig:MCvsNUMSRvsINTR}
\end{figure}

\begin{figure}
\centering
\includegraphics[scale=0.4]{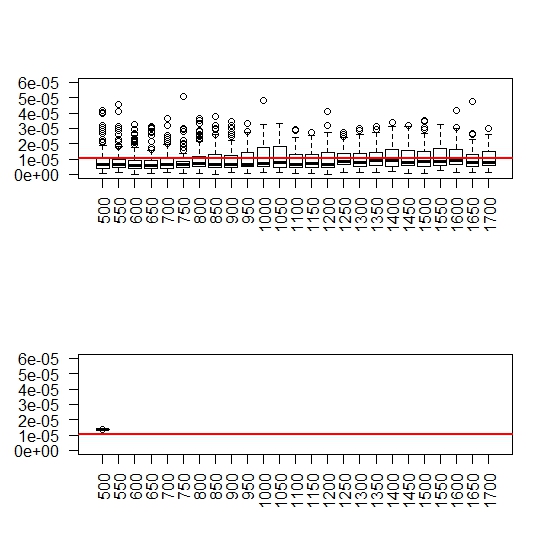}
\caption{Boxplots of 100 replications of the procedure based on Monte Carlo integration (above) and sums of Riemann (below) which approximates the Fisher information matrix of the model $0.25\mathcal{N}(-10,1)+0.10\mathcal{N}(0,5)+0.65\mathcal{N}(15,7)$ for sample sizes from $500$ to $1700$. The value obtained via numerical integration is represented by the red line (in the graph below, all the approximations obtained with more than $550$ knots give the same result, exactly equal to the one obtained via Gauss-Kronrod quadrature).}
\label{fig:MCvsRIEMbxp}
\end{figure}

\begin{figure}
\centering
\includegraphics[scale=0.4]{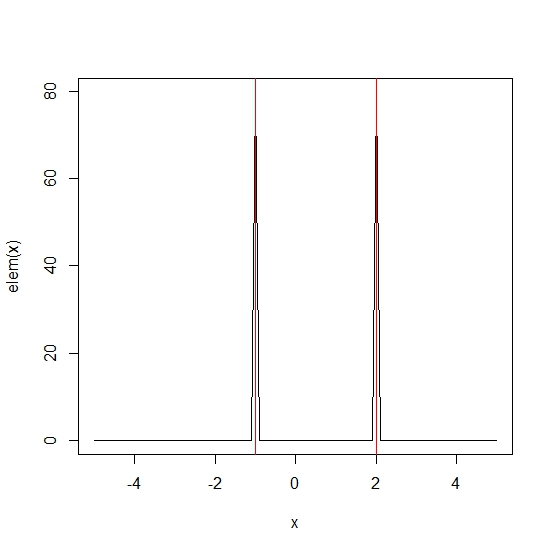}
\caption{The first element on the diagonal of the Fisher information matrix relative to the first weight of the two-component Gaussian mixture model $0.5 \mathcal{N}(-1,0.01)+0.5 \mathcal{N}(2,0.01)$.}
\label{fig:FishInfoelem}
\end{figure}

Thus, we have decided to use the sums of Riemann (with a number of points equal to $550$) to approximate the Jeffreys prior when the standard deviations are sufficiently big and Monte Carlo integration (with sample sizes of $1500$) when they are too small. In this case, the variability of the results seems to decrease as $\sigma_i$ approaches $0$, as shown in Figure \ref{fig:MCsmallsd}.
 
\begin{figure}
\centering
\includegraphics[scale=0.4]{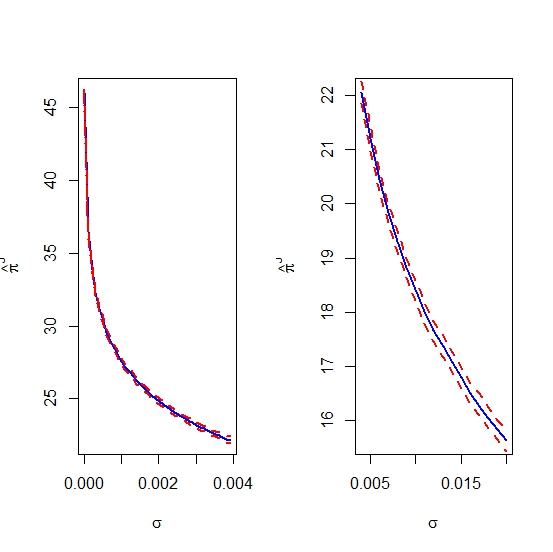}
\caption{Approximation of the Jeffreys prior (in log-scale) for the two-component Gaussian mixture model $0.5 \mathcal{N}(-1,\sigma)+0.5\mathcal{N}(2,\sigma)$, where $\sigma$ is taken equal for both components and decreasing.}
\label{fig:MCsmallsd}
\end{figure}

We have chosen to consider Monte Carlo samples of size equal to $1500$ because both the value of the approximation and
its standard deviations are stabilizing.

An adaptive MCMC algorithm has been used to define the variability of the
kernel density functions used to propose the moves. During the burnin, the
variability of the kernel distributions has been reduced or increased depending
on the acceptance rate, in a way such that the acceptance rate stay between $20\%$
and $40\%$. The transitional kernel used have been truncated normals for the
weights, normals for the means and log-normals for the standard deviations (all
centered on the values accepted in the previous iteration).

\vspace{0.3cm}
\section{Conclusion}\label{sec:concl}

This thorough analysis of the Jeffreys priors in the setting of Gaussian mixtures shows that mixture distributions can
also be considered as an ill-posed problem with regards to the production of non-informative priors. Indeed, we have
shown that most configurations for Bayesian inference in this framework do not allow for the standard Jeffreys prior to
be taken as a reference.  While this is not the first occurrence where Jeffreys priors cannot be used as reference
priors, the wide range of applications of mixture distributions weights upon this discovery and calls for a  
new paradigm in the construction of non-informative Bayesian procedures for mixture inference. Our proposal in Section
\ref{sec:alternative} could constitute such a reference, as it simplifies the representation of
\cite{mengersen:robert:1996}. 

\bibliographystyle{ims}  
\hyphenation{Post-Script Sprin-ger}

\end{document}